\documentclass[1pt,a4paper]{article}
\usepackage[english]{babel}
\usepackage{color}
\usepackage{url}
\usepackage[english]{babel}
\usepackage{color}
\usepackage{graphicx}
\usepackage{hyperref}
\usepackage{amsmath,amssymb,amsthm}
\usepackage[T1]{fontenc}
\usepackage[utf8]{inputenc}
\usepackage{authblk}
\usepackage[english]{babel}
\usepackage{amsmath,amssymb,amsthm}
\usepackage{color}
\usepackage{graphics}
\usepackage{amsmath,amssymb,amsthm}
\newtheorem{definition}{Definition}
\newtheorem{theorem}[definition]{Theorem}
\newtheorem{lemma}[definition]{Lemma}
\newtheorem{corollary}[definition]{Corollary}
\newtheorem{remark}[definition]{Remark}
\newtheorem{proposition}[definition]{Proposition}
\newtheorem{example}{Example}

\renewcommand{\tilde}[1]{\widetilde{#1}}

\newcommand{\bbF}{{\mathbb{F}}}

\title{An improvement of the Feng-Rao bound for primary codes}
\author[1]{Olav Geil\thanks{olav@math.aau.dk}}
\author[1,2]{Stefano Martin\thanks{stefano@math.aau.dk}}
\affil[1]{Department of Mathematical Sciences, Aalborg University}
\affil[2]{Engineering Software Institute, East China Normal University}

\begin{document}
\maketitle
\begin{abstract}
We present a new bound for the minimum distance of a general primary
linear code. For affine variety codes defined from generalised
$C_{ab}$ polynomials the new bound often
improves dramatically on the Feng-Rao bound for primary
codes~\cite{AG,geithom}. The method does not only work for the minimum distance but
can be applied to any generalised Hamming weight.\\


\end{abstract}

\section{Introduction}
In this paper we present an improvement to the Feng-Rao bound for
{\textit{primary}} codes~\cite{AG,geithom,agismm}. Our method does not only apply to the minimum
distance but estimates any generalised Hamming weight. In the same way
as the Feng-Rao bound for primary codes suggests an improved code
construction our new bound does also. The new bound
is particular suited for affine variety codes for which it often improves
dramatically on the Feng-Rao bound. Interestingly, for such codes it can
be viewed as a simple application of the footprint bound from Gr\"{o}bner
basis theory. We pay particular attention to the case of the affine
variety being defined by a bivariate polynomial that, in the support, has two univariate
monomials of the same weight and all other monomials of lower weight. Such polynomials can be viewed as a
generalisation of the polynomials defining $C_{ab}$ curves and 
therefore we name them {\em{generalised $C_{ab}$ polynomials}}. We
develop a method for constructing generalised $C_{ab}$ polynomials with many
zeros by the use of
$({\mathbb{F}}_{p^m},{\mathbb{F}}_p)$-polynomials, that are polynomials returning values in ${\mathbb{F}}_p$ when evaluated in ${\mathbb{F}}_{p^m}$ (see, \cite[Chap.\ 1]{redei}). Here, $p$ is
any prime power and $m$ is an integer larger than $1$. With this
method in hand we can design long 
affine variety codes for which our bound produces good results. The new
bound of the present paper is closely related to an improvement of the
Feng-Rao bound for {\textit{dual}} codes that we presented recently
in~\cite{geilmartin2013further}. Recall from~\cite{agismm} that the usual
Feng-Rao bound for primary and dual codes can be viewed as
consequences of each other. This result holds when one uses the
concept of well-behaving pairs or one-way well-behaving pairs. For
weakly well-behaving pairs a possible connection is unknown. In a
similar way as the proof from~\cite{agismm} breaks down for weakly
well-behaving, it also breaks down when one tries to establish a connection between the new bound from the present paper and the new bound from~\cite{geilmartin2013further}. 
We shall leave it as an open problem to decide if the two
bounds are consequences of each other or not. \\ 

In the first part of the paper we concentrate solely on affine variety
codes. For such codes the new method is intuitive. We start by
formulating in Section~\ref{sectwo} our new bound at the level of
affine variety codes and explain how it gives rise to an improved code construction $\tilde{E}_{imp}(\delta)$.
 Then we continue
in Section~\ref{secthree} by showing how to construct generalised
$C_{ab}$ polynomials with many zeros. In Section~\ref{secfour} we
give a thorough treatment of codes defined from so-called optimal
generalised $C_{ab}$ polynomials demonstrating the strength of our new
method. In Section~\ref{secfourandaquarter} we show how to
improve the improved code construction $\tilde{E}_{imp}(\delta)$ even further. This is done for the case of the affine variety being the Klein quartic. Having up till now
only considered the minimum distance, in Section~\ref{secfourandahalf}
we explain how to deal with generalised Hamming weights. 
Then we turn to the level of general
primary linear codes lifting in Section~\ref{secfive} our method to a
 bound on any primary linear code. In Section~\ref{secsix} we recall the recent bound
from~\cite{geilmartin2013further} on dual codes, and in Section~\ref{seccomp} we discuss the relation between this bound and the new bound of the present paper. Section~\ref{seceight} is the conclusion.
\section{Improving the Feng-Rao bound for primary affine variety
  codes}~\label{sectwo}
Affine variety codes were introduced by Fitzgerald and Lax
in~\cite{lax} as follows. For $q$ a prime power consider an ideal $I \subseteq
{\mathbb{F}}_q[X_1, \ldots , X_m]$ and define 
\begin{equation}
I_q=I + \langle X_1^q-X_1, \ldots , X_m^q-X_m\rangle,\label{eqiq}
\end{equation}
$$R_q={\mathbb{F}}_q[X_1, \ldots , X_m]/I_q.$$
Let $\{P_1, \ldots , P_n\}={\mathbb{V}}_{\mathbb{F}_q}(I_q)$ be the
corresponding variety over ${\mathbb{F}}_q$. Here, $P_i \neq P_j$ for
$i \neq j$. Define the ${\mathbb{F}}_q$-linear map ${\mbox{ev}} : R_q
\rightarrow {\mathbb{F}}_q^n$ by ${\mbox{ev}}(A+I_q)=(A(P_1), \ldots ,
A(P_n))$. It is well-known that this map is a vector space
isomorphism.
\begin{definition}\label{defaffcode}
Let $L$ be an ${\mathbb{F}}_q$ vector subspace of $R_q$. Define
$C(I,L)={\mbox{ev}}(L)$ and $C^\perp(I,L)=\big(C(I,L)\big)^\perp$. 
\end{definition}
We
shall call $C(I,L)$ a primary affine variety code and $C^\perp(I,L)$ a
dual affine variety code. For the case of primary affine variety codes both
the Feng-Rao bound and the bound of the present paper can be viewed as
consequences of the footprint bound from Gr\"{o}bner basis theory
as we now explain. 
\begin{definition}\label{deffoot}
Let $J \subseteq k[X_1, \ldots ,X_m]$ be an ideal and let $\prec$ be a
fixed monomial ordering. Here, $k$ is an arbitrary field. Denote by
${\mathcal{M}}(X_1, \ldots , X_m)$ the monomials in the variables
$X_1,\ldots , X_m$. The footprint of $J$ with respect to $\prec$ is
the set
\begin{eqnarray}
\Delta_{\prec}(J)&=&\{ M \in {\mathcal{M}}(X_1, \ldots , X_m) \mid M
{\mbox{ is not }}\nonumber \\
&&{\mbox{ \ \ \ \ \ the leading monomial of any polynomial in }} J\}.\nonumber
\end{eqnarray} 
\end{definition}
\begin{proposition}\label{probasis}
Let the notation be as in Definition~\ref{deffoot}. The set $\{M+J \mid M \in \Delta_\prec (J)\}$ constitutes a basis for
$k[X_1, \ldots , X_m]/J$ as a vector space over $k$. 
\end{proposition}
\begin{proof}
See~\cite[Pro.\ 4, Sec.\ 5.3]{clo}.
\end{proof}
We shall make extensive use of the following incidence of the
footprint bound (for a more general version, see~\cite{geil2000footprints}).
\begin{corollary}\label{thefoot}
Let $F_1, \ldots , F_s \in {\mathbb{F}}_q[X_1, \ldots , X_m]$. For any
monomial ordering $\prec$ the variety
${\mathbb{V}}_{\mathbb{F}_q}(\langle F_1, \ldots ,F_s\rangle)$ is of
size equal to $\# \Delta_\prec (\langle F_1, \ldots , F_s, X_1^q-X_1,
\ldots , X_m^q-X_m\rangle )$. 
\end{corollary}
\begin{proof}
Follows from Proposition~\ref{probasis} and the fact that the map
${\mbox{ev}}$ is a bijection.
\end{proof}
We next recall the interpretation from~\cite{bookAG} of the Feng-Rao
bound for primary affine variety codes.
\begin{definition}\label{deftksh1}
A basis $\{B_1+I_q, \ldots , B_{\dim (L)}+I_q\}$ for a subspace $L \subseteq R_q$
where ${\mbox{Supp}}(B_i) \subseteq \Delta_\prec (I_q)$ for $i=1,
\ldots , \dim(L)$ and where ${\mbox{lm}}(B_1) \prec \cdots \prec
{\mbox{lm}}(B_{\dim(L)})$, is said to be well-behaving with respect to
$\prec$. Here, ${\mbox{lm}}(F)$ means the leading monomial of the polynomial $F$.
\end{definition}
For fixed $\prec$ the sequence $({\mbox{lm}}(B_1), \ldots
,{\mbox{lm}}(B_{\dim(L)}))$ is the same for all choices of
well-behaving bases of $L$. Therefore the following definition makes sense.
\begin{definition}\label{deffirkant}
Let $L$ be a subspace of $R_q$ and define
$$\Box_\prec(L)=\{ {\mbox{lm}}(B_1), \ldots
,{\mbox{lm}}(B_{\dim(L)})\},$$
where $\{B_1+I_q, \ldots , B_{\dim(L)}+I_q\}$ is any well-behaving
basis for $L$.
\end{definition}
The concept of one-way well-behaving plays a crucial role in the
Feng-Rao bound as well as in our new bound. It is a relaxation of the
well-behaving property and the weakly well-behaving property
(see~\cite{bookAG,geithom} for a reference) and 
therefore it gives the strongest bounds.
\begin{definition}\label{defaffowb}
Let ${\mathcal{G}}$ be a Gr\"{o}bner basis for $I_q$ with respect to
$\prec$. An ordered pair of monomials $(M_i,M_j)$, $M_i, M_j \in
\Delta_\prec(I_q)$ is said to be one-way well-behaving (OWB) if for
all $H \in {\mathbb{F}}_q[X_1, \ldots , X_m]$ with ${\mbox{Supp}}(H)
\subseteq \Delta_\prec(I_q)$ and ${\mbox{lm}}(H)=M_i$ it holds that
$${\mbox{lm}}(M_iM_j {\mbox{ rem }} {\mathcal{G}})={\mbox{lm}}(HM_j
{\mbox{ rem }} {\mathcal{G}}).$$
Here, $F {\mbox{ rem }} {\mathcal{G}}$ means the remainder of $F$
after division with ${\mathcal{G}}$ (see~\cite[Sec.\ 2.3]{clo} for the
division algorithm for multivariate polynomials).
\end{definition}
As noted in~\cite{bookAG} the concept of OWB is independent of which
Gr\"{o}bner basis ${\mathcal{G}}$ is used as long as $I_q$ and $\prec$
are fixed. We are now ready to describe the Feng-Rao bound for primary
affine variety codes. We include the proof from~\cite[Th.\
4.9]{bookAG}.
\begin{theorem}\label{theseven}
Let ${\mathcal{G}}$ be a Gr\"{o}bner basis for $I_q$ with respect to
$\prec$.  Consider a non-zero word $\vec{c}$ and let 
$A$
be the unique
polynomial such that ${\mbox{Supp}}(A) \subseteq \Delta_\prec(I_q)$
and 
$\vec{c}={\mbox{ev}}(A)$. Let ${\mbox{lm}}(A)=P$. We have
\begin{eqnarray}
w_H(\vec{c})&\geq & \# \{ K \in \Delta_\prec(I_q) \mid \exists N \in
\Delta_\prec (I_q) {\mbox{ such that }} \nonumber \\
&&{\mbox{ \ \ \ }} (P,N) {\mbox{ is OWB and }} {\mbox{lm}}(PN {\mbox{
    rem }} {\mathcal{G}})=K\}.\label{eqafffr}
\end{eqnarray}
A bound on the minimum distance of $C(I,L)$ is found by
taking the minimum of~(\ref{eqafffr}) when  $P$ runs through $\Box_\prec(L)$.
\end{theorem}
\begin{proof}
From Corollary~\ref{thefoot} we know that 
\begin{eqnarray}
w_H(\vec{c}) &=& n - \#
\Delta_\prec (I_q+\langle A \rangle )\nonumber \\
&=&\#\Delta_{\prec}(I_q)-\#
\Delta_\prec(I_q+\langle A \rangle )\nonumber \\
&=&\# \bigg(\Delta_\prec(I_q) \backslash \Delta_\prec(I_q+\langle A
\rangle)\bigg).\label{eqerher}
\end{eqnarray}
If $N,K \in \Delta_\prec
(I_q)$ satisfy that $(P,N)$ is OWB and ${\mbox{lm}}(PN {\mbox{ rem }}
{\mathcal{G}} )=K$ then $K \in \Delta_\prec(I_q) \backslash
\Delta_\prec(I_q +\langle A \rangle)$. Hence,
\begin{eqnarray}
w_H(\vec{c})&\geq&\# \{ K \in \Delta_\prec(I_q) \mid \exists N \in \Delta_\prec (I_q)
\nonumber \\
&&{\mbox{ \ \ \ \ \ \ \ \ such that }} (P,N) {\mbox{ is OWB and }}
{\mbox{lm}}(PN {\mbox{ rem }} {\mathcal{G}}) = K\}.\nonumber
\end{eqnarray}
\end{proof}
The Feng-Rao bound is particular suited for affine varieties which
satisfy the order domain conditions~\cite[Def.\ 4.22]{bookAG}. For
other varieties it does not seem to produce very good results. The new bound
of the present paper solves this problem for affine varieties which
satisfy the first half of the order domain conditions. This gives a lot of
freedom as the latter set of varieties is much larger than the
former. In its most general form the order domain
conditions involves a weighted degree monomial ordering with weights
in ${\mathbb{N}}_0^r \backslash \{ \vec{0}\}$, $r$ a positive integer (see \cite[Def.\
4.21]{bookAG}). Here, for simplicity we shall only consider weights in
$\mathbb{N}$. 
\begin{definition}\label{defwdeg}
Let $w(X_1), \ldots , w(X_m) \in {\mathbb{N}}$ and define the weight
of $X_1^{i_1}\cdots X_m^{i_m}$ to be the number $w(X_1^{i_1} \cdots X_m^{i_m})=i_1 w(X_1)+ \cdots +i_m w(X_m)$. The weighted degree
ordering $\prec_w$ on ${\mathcal{M}}(X_1, \ldots  ,X_m)$ is the
ordering with $X_1^{i_1} \cdots X_m^{i_m} \prec_w X_1^{j_1} \cdots
X_m^{j_m}$ if either $w(X_1^{i_1} \cdots X_m^{i_m}) < w(X_1^{j_1} \cdots
X_m^{j_m})$ holds or $w(X_1^{i_1} \cdots X_m^{i_m}) = w(X_1^{j_1} \cdots
X_m^{j_m})$ holds but $X_1^{i_1} \cdots X_m^{i_m} \prec^\prime X_1^{j_1} \cdots
X_m^{j_m}$. Here, $\prec^\prime$ is some fixed monomial ordering. When $\prec^\prime$ is the lexicographic ordering $\prec_{lex}$ with
$X_m \prec_{lex} \cdots \prec_{lex} X_1$ we shall call $\prec_w$ a weighted degree lexicographic ordering.
\end{definition}
We now state the order domain conditions which play a central role in
the present paper. 
\begin{definition}\label{defupher}
Consider an ideal  $J \subseteq k[X_1, \ldots , X_m]$ where $k$ is a
field. Let a weighted degree ordering $\prec_w$ be given. Assume that $J$ possesses a Gr\"{o}bner basis
${\mathcal{F}}$ with respect to $\prec_w$ such
that:
\begin{itemize}
\item[(C1)] Any $F \in {\mathcal{F}}$ has exactly two monomials of highest
weight.
\item[(C2)] No two monomials in $\Delta_{\prec_w}(J)$ are of the
same weight. 
\end{itemize}
Then we say that $J$ and $\prec_w$ satisfy the order
domain conditions. 
\end{definition}
In the following we restrict to weighted degree orderings where $\prec^\prime =\prec_{lex}$. That is, $\prec_w$ shall always be a weighted degree lexicographic ordering.
\begin{example}\label{exord}
Consider $I=\langle X^2+X-Y^3 \rangle \subseteq
{\mathbb{F}}_4[X,Y]$ and $I_4$ accordingly (see~(\ref{eqiq})). Choosing $X=X_1$, $Y=X_2$,
$w(X)=3$ and $w(Y)=2$ we see that the order domain conditions are
satisfied. By inspection we have $$\Delta_{\prec_w}(I_4)=\{1, Y, X,
Y^2, XY, Y^3, XY^2, XY^3\}$$ with corresponding weights $\{0, 2, 3, 4,
5, 6, 7, 9\}$. 
Consider a word $\vec{c}={\mbox{ev}}(A+I_4)$ where $A=a_1
1+a_2Y+a_3X$, $a_1, a_2 \in {\mathbb{F}}_4$ and
$a_3 \in {\mathbb{F}}_4 \backslash \{0\}$. By Corollary~\ref{thefoot} the length is $n=8$. We now estimate
the Hamming weight $w_H(\vec{c})=\# \big( \Delta_{\prec_w} (I_4)
\backslash \Delta_{\prec_w}(I_4 + \langle A \rangle)\big)$ (see~(\ref{eqerher})). The
following elements in $\Delta_{\prec_w}(I_4)$ do not belong to
$\Delta_{\prec_w}(I_4 +\langle A\rangle)$. Namely, ${\mbox{lm}}(A
\cdot 1)=X$, ${\mbox{lm}}(A
\cdot Y)=XY$, ${\mbox{lm}}(A
\cdot Y^2)=XY^2$, ${\mbox{lm}}(A
\cdot Y^3)=XY^3$, and ${\mbox{lm}}(A
\cdot X {\mbox{ rem }} X^2+X-Y^3)=Y^3$. Observe that the last
calculation 
holds due to the fact that $X^2+X-Y^3$ contains exactly two monomials
of the highest weight. We have shown that the Hamming weight of
$\vec{c}$ is
at least $5$. With the proof of Theorem~\ref{theseven} in mind an
equivalent formulation of the above is to observe that
$(X,1)$, $(X,Y)$, $(X,Y^2)$, $(X,Y^3)$, and $(X,X)$ are OWB. Another
equivalent method is guaranteed by the condition that
$\Delta_{\prec_w}(I)$ does not contain two monomials of the same
weight. This implies that rather than counting the above OWB pairs we
only need to observe that $w(\Delta_{\prec_w}(I_4))  \cap \big(
w(X)+w(\Delta_{\prec_w}(I_4)) \big)=\{3,5,6,7,9\}$. Again, a set of
size $5$.
\end{example}
The following Proposition (corresponding to~\cite[Pro.\ 4.25]{bookAG})
summarises how the Feng-Rao bound is supported by the order domain condition.
\begin{proposition}\label{proptotwo}
Assume $I \subseteq {\mathbb{F}}_q[X_1, \ldots , X_m]$ and $\prec_w$
satisfy the order domain conditions. Consider $I_q=I+\langle
X_1^q-X_1, \ldots , X_m^q-X_m\rangle$. A pair $(P,N)$ where $P, N \in
\Delta_{\prec_w}(I_q)$ is OWB if $w(P)+w(N) \in w(\Delta_{\prec_w}
(I_q))$. 
\end{proposition}
The order domain conditions
historically~\cite{handbook,pellikaan2001existence,AG,bookAG} were designed to
support the Feng-Rao bounds and therefore it is not surprising that the
bound does not work very well without them. The improvement to the
Feng-Rao bound that we introduce below allows us to consider relaxed
conditions in that we can produce good estimates in the case that the
order domain condition (C1) is satisfied but (C2) is not. The following example illustrates the idea in our improvement to
Theorem~\ref{theseven}.
\begin{example}\label{exmot}
Consider $I=\langle X^4+X^2+X-Y^6-Y^5-Y^3\rangle \subseteq
{\mathbb{F}}_8[X,Y]$. 
Let $\prec_w$ be the weighted
degree lexicographic ordering (Definition~\ref{defwdeg}) given by
$X=X_1$, $Y=X_2$, $w(X)=3$ and $w(Y)=2$. From~\cite[Sec.\ 3]{salazar} and \cite[Sec.\
4.2]{geilmartin2013further} we know that the variety
${\mathbb{V}}_{\mathbb{F}_8}(I_8)$ is of size $32$. 
Combining this observation with 
Corollary~\ref{thefoot} we see that 
$$\Delta_{\prec_w} (I_8) = \{ X^\alpha Y^\beta  \mid 0 \leq \alpha < 4, 0
\leq \beta < 8\}.$$
By inspection we see that some weights appear twice in $\Delta_{\prec_w}(I_8)$, some only once.  
Consider $\vec{c}={\mbox{ev}}(A+I_8)$ where ${\mbox{lm}}(A)=X^3$. That
is,
\begin{eqnarray}
A&=&a_1 1+a_2Y+a_3 X+a_4 Y^2 + a_5
XY+a_6 Y^3+a_7 X^2\nonumber \\
&&+ a_8 XY^2+a_9 Y^4+a_{10}
X^2Y+a_{11}XY^3+a_{12}X^3.\nonumber
\end{eqnarray}
Here, $a_i \in {\mathbb{F}}_8$, $i=1, \ldots , 12$ and
$a_{12} \neq 0$. Note that $A$ has two monomials of the
highest weight if $a_{11} \neq 0$, namely $X^3$ and $XY^3$. Following the proof of
Theorem~\ref{theseven} we consider $P=X^3$ and look for $N,K \in
\Delta_{\prec_w}(I_8)$ such that $(P,N)$ is OWB and ${\mbox{lm}}(PN
{\mbox{ rem }} {\mathcal{G}})=K$. We have the following possible
choices of $(N,K)$, namely $(1,X^3)$, $(Y,X^3Y)$, $(Y^2,X^3Y^2), \ldots
, (Y^7,X^3Y^7),(X^3,X^2Y^6),(X^3Y,X^2Y^7)$. From this we conclude that $w_H(\vec{c}) \geq
10$.\\ 
Note that $X^3 \cdot X {\mbox{ rem }} {\mathcal{G}}=Y^6$. However,
$(X^3,X)$ is not OWB as 
\begin{equation}
XY^3 \prec_w X^3 {\mbox{ but }} XY^3 \cdot X {\mbox{ rem }} {\mathcal{
    G}} = X^2Y^3 \succ_w Y^6. \label{eqsnabel}
\end{equation}
Our improved method consists in considering separately two
different cases: $XY^3 \in {\mbox{Supp}}(A)$ and $XY^3 \notin
{\mbox{Supp}}(A)$.\\

\noindent {\underline{Case 1:}} Assume $a_{11}\neq
0$. Following~(\ref{eqsnabel}) we see that ${\mbox{lm}}(A \cdot X
{\mbox{ rem }} {\mathcal{G}})=X^2Y^3$. In a similar way we derive ${\mbox{lm}}(A \cdot XY
{\mbox{ rem }} {\mathcal{G}})=X^2Y^4$
and ${\mbox{lm}}(A \cdot XY^2
{\mbox{ rem }} {\mathcal{G}})=X^2Y^5$. From this we conclude
$$\Delta_{\prec_w}(I_q+\langle A \rangle)\subseteq \{ X^\alpha Y^\beta
\mid 0 \leq \alpha < 3, 0 \leq \beta <8, {\mbox{ and if }} \alpha=2
{\mbox{ then }} \beta < 3\}$$
and therefore that $w_H(\vec{c}) \geq n-\#
\Delta_{\prec_w}(I_8+\langle A \rangle)=32-19=13$.\\

\noindent {\underline{Case 2:}} Assume $a_{11}=0$. 
This means that we do not have to worry about~(\ref{eqsnabel}) and consequently
${\mbox{lm}}(A\cdot X {\mbox{ rem }} {\mathcal{G}} )=Y^6$ holds. In a similar way we derive ${\mbox{lm}}(A\cdot X^2 {\mbox{ rem }} {\mathcal{G}} )=XY^6$,
${\mbox{lm}}(A\cdot XY {\mbox{ rem }} {\mathcal{G}} )=Y^7$, and
${\mbox{lm}}(A\cdot X^2Y {\mbox{ rem }} {\mathcal{G}} )=XY^7$. We
conclude that
$$\Delta_{\prec_w}(I_q+\langle A \rangle ) \subseteq \{X^\alpha Y^\beta
\mid 0 \leq \alpha < 3, 0 \leq \beta < 6\}$$
and therefore from the proof of Theorem~\ref{theseven} we have that
$w_H(\vec{c}) \geq n-\# \Delta_{\prec_w}(I_8+\langle A \rangle )=32-18=14$.\\
 
\noindent In conclusion $w_H(\vec{c}) \geq \min \{13,14\}=13$.
\end{example}
With Example~\ref{exmot} in mind we now improve upon
Theorem~\ref{theseven}.
\begin{definition}\label{defstrong}
Let ${\mathcal{G}}$ be a Gr\"{o}bner basis for $I_q$ with respect to a fixed arbitrary monomial ordering 
$\prec$. Write $\Delta_\prec(I_q) =\{ M_1, \ldots , M_n\}$ with $M_1
\prec \cdots \prec M_n$. Let ${\mathcal{I}}=\{ 1, \ldots , n\}$ and
consider ${\mathcal{I}}^\prime \subseteq {\mathcal{I}}$. An ordered
pair of monomials $(M_i,M_j)$, $ 1 \leq i, j \leq n$ is said to be
strongly one-way well-behaving (SOWB) with respect to
${\mathcal{I}}^\prime$ if for all $H$ with ${\mbox{Supp}}(H) \subseteq
\{M_s \mid s \in {\mathcal{I}}^\prime \}$, $M_i \in {\mbox{Supp}}(H)$
it holds that ${\mbox{lm}}(M_iM_j {\mbox{ rem }}
{\mathcal{G}})={\mbox{lm}}(H M_j {\mbox{ rem }} {\mathcal{G}})$. 
\end{definition}
In the following, when writing $\Delta_\prec(I_q) =\{ M_1, \ldots , M_n\}$, we shall always assume that  $M_1
\prec \cdots \prec M_n$ holds.\\
Consider a non-zero codeword $\vec{c}={\mbox{ev}}(A+I_q)$, where
$A=\sum_{s=1}^i a_sM_s$, $i \geq 2$, $a_s \in {\mathbb{F}}_q$ for $s=1,
\ldots ,i$ and $a_i \neq 0$. Let $v$ be an integer $1 \leq v <
i$. We consider $v+1$ different cases that cover all possibilities:\\

\noindent {\underline{Case 1:}} $a_{i-1}\neq 0$.\\

\noindent {\underline{Case 2:}} $a_{i-1}= 0$, $a_{i-2}\neq 0$.\\

 \ \ \ $\ \vdots$\\

\noindent {\underline{Case v:}} $a_{i-1}=a_{i-2}= \cdots
=a_{i-v+1}= 0$, $a_{i-v}\neq 0$.\\

\noindent {\underline{Case v+1:}} $a_{i-1}=\cdots =
a_{i-v}=0$.\\

\noindent For each of the above $v+1$ cases we shall estimate $n-\#
\Delta_{\prec}(I_q+\langle A \rangle )$. Then the minimal obtained value
constitutes a lower bound on $w_H(\vec{c})$. Note that in
Example~\ref{exmot} 
we used $v=1$.
\begin{theorem}\label{thenewbound}
Let $\prec$ be a fixed arbitrary monomial ordering. Consider $\vec{c}={\mbox{ev}}(\sum_{s=1}^i a_s M_s +I_q)$,
$a_s \in {\mathbb{F}}_q$, $s=1, \ldots , i$, and $a_i \neq
0$. Let $v$ be an integer $0 \leq v < i$. We have 
$$w_H(\vec{c}) \geq \min \{ \#{\mathcal{L}}(1), \ldots ,
\#{\mathcal{L}}(v+1)\}$$
 where for $t=1, \ldots , v$ we define ${\mathcal{L}}(t)$ as follows:
\begin{eqnarray}
{\mathcal{L}}(1)&=& \big\{ K \in \Delta_\prec(I_q) \mid \exists M_j \in \Delta_\prec(I_q)
{\mbox{ such that either }}\nonumber \\
&&(M_i,M_j) {\mbox{ is OWB  and }} {\mbox{lm}}(M_iM_j {\mbox{ rem }} {\mathcal{G}})=K
{\mbox{ or }}
\nonumber \\
&&(M_{i-1},M_j) {\mbox{ is SOWB with respect to }} \{1, \ldots ,
i\}\nonumber \\
&&{\mbox{ and }} {\mbox{lm}}(M_{i-1}M_j {\mbox{ rem }} {\mathcal{G}})=K\big\},
\nonumber \\
\ \nonumber \\
{\mathcal{L}}(2)&=& \big\{ K \in \Delta_\prec(I_q) \mid \exists M_j \in \Delta_\prec(I_q)
{\mbox{ such that either }}\nonumber \\
&&(M_i,M_j) {\mbox{ is SOWB with respect to }} \{ 1, \ldots , i-2,i\} \nonumber \\
&&{\mbox{ and }} {\mbox{lm}}(M_iM_j {\mbox{ rem }} {\mathcal{G}})=K
{\mbox{ or }}
\nonumber \\
&&(M_{i-2},M_j) {\mbox{ is SOWB with respect to }} \{1, \ldots ,i-2,
i\}\nonumber \\
&&{\mbox{ and }} {\mbox{lm}}(M_{i-2}M_j {\mbox{ rem }} {\mathcal{G}})=K\big\},
\nonumber \\
\ \nonumber \\
&&{\mbox{ \ \ \ \ \ \ \ \ \ \ \ \ \ \ \ \ \ \ \ \ \ \  \ \ }} \vdots \nonumber \\
\ \nonumber \\
{\mathcal{L}}(v)&=&\big\{ K \in \Delta_\prec(I_q) \mid \exists M_j \in \Delta_\prec(I_q)
{\mbox{ such that either }}\nonumber \\
&&(M_i,M_j) {\mbox{ is SOWB with respect to }} \{ 1, \ldots , i-v,i\} \nonumber \\
&&{\mbox{ and }} {\mbox{lm}}(M_iM_j {\mbox{ rem }} {\mathcal{G}})=K
{\mbox{ or }}
\nonumber \\
&&(M_{i-v},M_j) {\mbox{ is SOWB with respect to }} \{1, \ldots ,i-v,
i\}\nonumber \\
&&{\mbox{ and }} {\mbox{lm}}(M_{i-v}M_j {\mbox{ rem }} {\mathcal{G}})=K\big\},
\nonumber 
\end{eqnarray}
Finally,
\begin{eqnarray}
{\mathcal{L}}(v+1)&=& \big\{ K \in \Delta_\prec(I_q) \mid \exists M_j \in \Delta_\prec(I_q)
{\mbox{ such that }} (M_i,M_j) {\mbox{ is SOWB}}
\nonumber \\
&& {\mbox{ with respect to }} \{1, \ldots ,
i-v-1,i\} {\mbox{ and }} {\mbox{lm}}(M_iM_j {\mbox{ rem }} {\mathcal{G}})=K \big\}.
\nonumber 
\end{eqnarray}
Given a code $C(I,L)$ write $\Box_\prec(L)=\{M_{i_1}, \ldots ,
M_{i_{\dim (L)}}\}$. 
A lower bound on the minimum distance is obtained by repeating the
above calculation for each $i \in \{i_1, \ldots , i_{\dim (L) }\}$. For
each choice of $i$ an appropriate value $v$ is chosen.
\end{theorem}
\begin{proof}
If $v=0$ then only the last set is present and this set
equals the set in~(\ref{eqafffr}). 
For $v>0$ the $v+1$ expressions correspond to the $v+1$ cases described prior to
the theorem (in the same order). The proof technique resembles the
arguments used in Example~\ref{exmot}.
\end{proof}
\begin{remark}\label{remtksh2}
Consider an ideal $I\subseteq {\mathbb{F}}_q[X_1, \ldots , X_m]$ and a corresponding weighted degree
lexicographic ordering $\prec_w$ such that the order domain condition (C1) is 
satisfied but (C2) is not. Let ${\mathcal{F}}$ be a Gr\"{o}bner basis
for $I$ with respect to $\prec_w$. Assume Theorem~\ref{thenewbound} is used to
estimate the Hamming weight of 
$\vec{c}={\mbox{ev}}(A+I_q)$ where ${\mbox{lm}}(A)=M_i$. A natural choice of $v$ is the unique
non-negative integer which satisfies
$w(M_i)=w(M_{i-1})= \cdots  =w(M_{i-v}) > w(M_{i-v-1})$. To see why this
choice of $v$ is natural, note that when reducing $A M_j$ modulo
${\mathcal{F}}$ the weight of the leading monomial remains the
same. Hence, the leading monomial of $AM_j {\mbox{ rem }}
{\mathcal{F}}$ can  not be equal to $M_t M_j {\mbox{ rem }}
{\mathcal{F}}$ for $t \leq i-v-1$. On the other hand as illustrated in
Example~\ref{exmot} this may happen
when $t \geq i-v$. For $I$ and $\prec_w$ such that both order domain
conditions are satisfied the above choice of $v$ is $v=0$ and
Theorem~\ref{thenewbound} therefore simplifies to the usual Feng-Rao bound Theorem~\ref{theseven} in this case.   
\end{remark}
Theorem~\ref{thenewbound} can be applied to any code
$C(I,L)$. However, it is not clear if
there is any advantage in considering other choices of $L$ than
$L={\mbox{Span}}_{\mathbb{F}_q} \{ {\mbox{ev}}(M_{i_1}+I_q), \ldots ,
{\mbox{ev}}(M_{i_k}+I_q)\}$. 
When $i_1=1, \ldots , i_k=k$ we shall
denote the corresponding code by $E(k)$. Observe that
Theorem~\ref{thenewbound} suggests an improved code construction as
follows. 
\begin{definition}\label{defimpcode}
Fix non-negative numbers $v_1, \ldots , v_n$ and calculate
for each $M_i$, $i=1, \ldots , n$ the number in
Theorem~\ref{thenewbound} where $v=v_i$. Call these number
$\tilde{\sigma}(i)$, $i=1, \ldots , n$. We define
$\tilde{E}_{imp}(\delta)$ to be the code with
$L={\mbox{Span}}_{\mathbb{F}_q}\{ {\mbox{ev}}(M_i+I_q) \mid
\tilde{\sigma}(i) \geq \delta \}$. 
\end{definition}
\begin{proposition}
The minimum distance of $\tilde{E}_{imp}(\delta)$ satisfies $d(\tilde{E}_{imp}(\delta)) \geq \delta$.
\end{proposition}
The above improved code construction is in
the spirit of Feng and Rao's work. When improved codes are constructed on the basis of the Feng-Rao bound, Theorem~\ref{theseven}, rather than on the basis of the improved bound of the present paper, Theorem~\ref{thenewbound}, the notation used is $\tilde{E}(\delta )$ (see~\cite[Def.\ 4.38]{bookAG}).  
In Section~\ref{secfourandaquarter} we shall see that
one can sometimes derive even further improved codes from Theorem~\ref{thenewbound} than $\tilde{E}_{imp}(\delta)$.\\

We conclude this section by noting that in a straight forward manner one can enhance the above bound to deal
also with generalised Hamming weights. We postpone the discussion of
the details 
to Section~\ref{secfourandahalf}.
\section{Generalised $C_{ab}$ polynomials}\label{secthree}
As mentioned in the previous section good candidates for our new bound
are affine variety codes where the order domain condition (C1) is
satisfied, but the order domain condition (C2) is not. A particular simple
class of curves that satisfy the order domain conditions are the
well-known $C_{ab}$ curves. They were introduced by Miura 
in~\cite{miura1993algebraic,miura1,miura1998linear} to facilitate the
use of the Feng-Rao bound for dual codes. In this section we introduce
generalised $C_{ab}$ polynomials which corresponds to allowing the same
weight to occur more than once in the footprint (condition (C2)). It should be stressed that we make no assumption that generalised $C_{ab}$ polynomials are irreducible as it has no implication for our analysis.\\
From~\cite[App.\
B and the lemma
at p.\ 1416]{miura1998linear} we have a complete characterisation
of $C_{ab}$ curves. We shall adapt the description in~\cite{matsumoto1998c_ab}
which is an English translation of Miura's results. From~\cite[Th.\
1]{matsumoto1998c_ab} we have:
\begin{theorem}\label{thetksh3}
Let $\bar{k}$ be the algebraic closure of a perfect field $k$, ${\mathcal{X}}
\subseteq \bar{k}^2$ be a possibly reducible affine algebraic set
defined over $k$, $x,y$ the coordinate of the affine plane
$\bar{k}^2$, and $a$, $b$ relatively prime positive integers. The
following two conditions are equivalent:
\begin{itemize}
\item ${\mathcal{X}}$ is an absolutely irreducible algebraic curve with exactly
  one $k$ rational place $Q$ at infinity, and the pole divisors of $x$
  and $y$ are $bQ$ and $aQ$, respectively.
\item ${\mathcal{X}}$ is defined by a bivariate polynomial of the form
\begin{equation}
\alpha_{a,0}x^a+\alpha_{0,b}y^b+\sum_{ib+ja<ab}\alpha_{i,j}x^iy^j,\label{eqryutaroh1}
\end{equation}
where $\alpha_{i,j} \in k$ for all $i,j$ and $\alpha_{a,0}$,
$\alpha_{0,b}$ are non-zero.
\end{itemize}
\end{theorem}
The definition of $C_{ab}$ curves given in the literature is that
of~(\ref{eqryutaroh1}). We recall the
following result from~\cite{miura1998linear}. We adapt the description
from~\cite[Cor.\ 3]{matsumoto1998c_ab}.
\begin{proposition}\label{propquotient}
Let $F(X,Y) \in k[X,Y]$ be a polynomial of the
form~(\ref{eqryutaroh1}), $Q$ a unique place at infinity of the
$C_{ab}$ curve defined by $F(X,Y)$. Then 
$$\{X^iY^j+\langle F(X,Y)\rangle \mid 
0 \leq i \leq a-1, 0 \leq j
\}$$
is a $k$-basis for $k[X,Y]/\langle F(X,Y)\rangle$ and the elements in the
basis have pairwise distinct discrete valuations at $Q$. If the
$C_{ab}$ curve is non-singular, then 
$$k[X,Y]/\langle F(X,Y) \rangle
={\mathcal{L}}( \infty Q)$$ 
and a basis of ${\mathcal{L}}(mQ)$ is
$$\{X^iY^j+\langle F(X,Y)\rangle \mid  0 \leq i \leq a-1, 
0 \leq j, 
ai+bj \leq m\}$$
for any non-negative integer $m$.
\end{proposition}
Let $w(X)$ and $w(Y)$, respectively, be minus the discrete valuation of
$x$ at $Q$ and minus the discrete valuation of $y$ at $Q$,
respectively. Consider the corresponding weighted degree lexicographic
ordering with $X=X_1$ and $Y=X_2$. If we combine (\ref{eqryutaroh1})
with the first part of Proposition~\ref{propquotient}
we see that
$C_{ab}$ curves satisfy the order domain conditions. Observe, that we
can consider the related affine variety codes $C(I,L)$ and $C^\perp(I,L)$ regardless
of the curve being non-singular or not. This point of view is taken
in~\cite[Sec.\ 4.2]{handbook}. If the curve is non-singular the
corresponding affine variety code description does not have an
algebraic geometric code counterpart. We now introduce generalised
$C_{ab}$ polynomials.
\begin{definition}\label{defpab}
Let $w(X)=\frac{b}{\gcd (a,b)}$ and $w(Y)=\frac{a}{\gcd (a,b)}$ where
$a$ and $b$ are two different positive integers. Given a field $k$, 
let  $F(X,Y)=X^a+\alpha Y^b + R(X,Y)\subseteq k[X,Y]$, $\alpha \in k
\backslash \{ 0 \}$, 
be such that all monomials in the support of $R$ have smaller weight
than $w(X^a)=w(Y^b)=\frac{ab}{\gcd(a,b)}$. Then $F(X,Y)$ is called a generalized $C_{ab}$ polynomial.
\end{definition}
Miura in~\cite[Sec.\ 4.1.4]{miura1993algebraic} treated the curves related to irreducible generalized $C_{ab}$ polynomials. Besides that we do not require the generalized $C_{ab}$ polynomials to be irreducible, our point of view is different from Miura's as we will use for the code construction the algebra ${\mathbb{F}}_q[X,Y]/\langle F(X,Y)\rangle$. For generalized $C_{ab}$ polynomials this algebra does not in general equal a space ${\mathcal{L}}(m_1P_1+\cdots +m_sP_s)$, $P_1, \ldots, P_s$ being rational places. We mention that the variations of $C_{ab}$ curves considered by Feng and Rao in~\cite{FR1} is different from Definition~\ref{defpab}.\\

For the code construction we would like to have
generalised $C_{ab}$ 
polynomials with many zeros and at the same time to  have a variety of
possible $a,b$ to choose from, as these parameters turn out to play a crucial role
in our bound for the minimum distance. As we shall now demonstrate
there is a simple technique for deriving this when the field under
consideration is not prime. The situation is in
contrast to $C_{ab}$ curves for which it is only known how to get many
points for restricted classes of $a$ and $b$. Our method builds on
ideas from~\cite{salazar} and \cite[Sec.\ 5]{miura1993algebraic}.\\
Let $p$ be a prime power and $q=p^m$ where $m\geq 2$ is an
integer. The technique that we shall employ involves letting $F(X,Y)=G(X)-H(Y)$ where both $G$ and
$H$ are  $({\mathbb{F}}_{p^m},{\mathbb{F}}_p)$-polynomials. 
\begin{definition}\label{deftksh4}
Let $m$ be an integer, $m \geq 2$. A polynomial $F(X) \in
{\mathbb{F}}_{p^m}[X]$ is called an
$({\mathbb{F}}_{p^m},{\mathbb{F}}_p)$-polynomial if $F(\gamma ) \in
{\mathbb{F}}_p$ holds for all $\gamma \in {\mathbb{F}}_{p^m}$.
\end{definition}
An obvious characterisation of
$({\mathbb{F}}_{p^m},{\mathbb{F}}_p)$-polynomials is that
$F(X)=(X^{p^m}-X)Q(X)+F^\prime (X)$, where $F^\prime(X)$ is an
$({\mathbb{F}}_{p^m},{\mathbb{F}}_p)$-polynomial of degree less than
$p^m$. Here, we used the convention that $\deg (0) =-\infty$. By Fermat's little theorem the set of
$({\mathbb{F}}_{p^m},{\mathbb{F}}_p)$-polynomials of degree less than
$p^m$ constitutes a vector space over ${\mathbb{F}}_p$. 
Clearly, one could derive a basis by Lagrange
interpolation. For our purpose, however, it is interesting to know what are the possible degrees of the polynomials in the vector space.
\begin{proposition}\label{propduderdu}
Let $C_{i_1}, \ldots , C_{i_t}$ be the different cyclotomic cosets modulo
$p^m-1$ (multiplication by $p$). Here, for $s=1, \ldots , t$ it is assumed that $i_s$ is chosen as the smallest element in the given
coset. For $s=1,
\ldots , t$,  $F_{i_s}(X)=\sum_{l \in C_{i_s}}X^l$, 
is an $({\mathbb{F}}_{p^m},{\mathbb{F}}_p)$-polynomial. Furthermore, the polynomial $X^{p^m-1}$ is an $({\mathbb{F}}_{p^m},{\mathbb{F}}_p)$-polynomial.\\
\end{proposition}
\begin{proof}
For all the polynomials $F$ in the proposition we have $F^p=F$.
\end{proof}

The set $\{F_{i_1}, \ldots , F_{i_t},X^{p^m-1}
\}$ contains two of the most prominent
$({\mathbb{F}}_{p^m},{\mathbb{F}}_p)$-polynomials, namely the trace polynomial
$F_1(X)=X^{p^{m-1}}+X^{p^{m-2}}+ \cdots + X^p+X$ and the
norm polynomial $X^{(p^m-1)/(p-1)}$. Note that the norm polynomial
equals $F_{(p^m-1)/(p-1)}$ if $p>2$. For $p=2$ it equals $X^{p^m-1}$. Observe also that
except for the constant polynomial $F_0=1$, the trace polynomial is of lowest possible degree.\\
From~\cite[Prop.\ 3.2]{hernando2011dimension} we have:
\begin{proposition}\label{prohern}
A polynomial $F(X)\in {\mathbb{F}}_{p^m}[X]$ is an
$({\mathbb{F}}_{p^m},{\mathbb{F}}_p)$-polynomial of degree less than
$p^m-1$ if and only if
$$F(X)= F_1(H(X)) {\mbox{ rem }} (X^{p^m-1}-1)$$
for some $H(X) \in {\mathbb{F}}_{p^m}[X]$.
\end{proposition}
From Proposition~\ref{propduderdu} and Proposition~\ref{prohern} we
conclude: 
\begin{corollary}
Let $F(X)$ be an $({\mathbb{F}}_{p^m},{\mathbb{F}}_p)$-polynomial of
degree less than $p^m$. Then $\deg(F) \in \{\deg(F_{i_1}), \ldots ,
\deg(F_{i_t}),p^m-1\}$.
\end{corollary}
We now return to the question of designing generalised $C_{ab}$
polynomials $F(X,Y)=G(X)-H(Y)$ with many zeros. One way of doing this is to
choose $G(X)$ to be the trace polynomial~\cite[Sec.\ 3]{salazar}. As is well-known this
polynomial maps exactly $p^{m-1}$ elements from ${\mathbb{F}}_{p^m}$ to each value in
${\mathbb{F}}_p$. Hence, such a polynomial $F(X,Y)$ must have $p^{2m-1}$
zeros. However, there are other polynomials in the above set with
properties similar to the trace polynomial.
\begin{proposition}\label{proptracelike}
Consider the polynomials $F_{i_s}$, $s=1, \ldots , t$ related to a field extension ${\mathbb{F}}_{p^m}/{\mathbb{F}}_p$, $m\geq 2$
(Proposition~\ref{propduderdu}). We have $\gcd (i_s ,p^m-1)=1$ if and only if
for each $\eta \in {\mathbb{F}}_p$ there exists exactly $p^{m-1}$
$\gamma \in {\mathbb{F}}_{p^m}$ such that $F_{i_s}(\gamma)=\eta$.
\end{proposition}
\begin{proof}
We have $F_{i_s}(X)=F_1(X^{i_s}) {\mbox{ mod }} (X^{q^m-1}-1)$, where
$F_1(X)$ is the trace polynomial. Under the condition that $\gcd (i_s,
p^m-1)=1$ the monomial $X^{i_s}$ defines a
bijective map from ${\mathbb{F}}_{p^m} \rightarrow
{\mathbb{F}}_{p^m}$. This proves the ``only if'' part. We leave the ``if'' part for the reader.
\end{proof}
\begin{example}\label{extksh5}
Consider first the field extension ${\mathbb{F}}_8/{\mathbb{F}}_2$. The
non-trivial cyclotomic cosets modulo $7$ are $C_1=\{1,2,4\}$, and
$C_3=\{3,6,5\}$. From this we find the following  
$({\mathbb{F}}_8,{\mathbb{F}}_2)$-polynomials: $F_1(X)=X^4+X^2+X$,
$F_3(X)=X^6+X^5+X^3$, and $X^7$. The first two polynomials have the
property described in Proposition~\ref{proptracelike}. This is a
consequence of $7$ being a prime.\\
Consider next the field extension
${\mathbb{F}}_{16}/{\mathbb{F}}_2$. The non-trivial 
cyclotomic cosets modulo $15$ are $C_1=\{1,2,4,8\}$,
$C_3=\{3,6,12,9\}$, $C_5=\{5,10\}$, $C_7=\{7,14,13,11\}$. Hence, we
get the following  $({\mathbb{F}}_{16},{\mathbb{F}}_2)$-polynomials $F_1(X)=X^8+X^4+X^2+X$, $F_3(X)=X^{12}+X^9+X^6+X^3$,
$F_5(X)=X^{10}+X^5$, $F_7(X)=X^{14}+X^{13}+X^{11}+X^7$, and
$X^{15}$. The polynomials with the property described in
Proposition~\ref{proptracelike} are $F_1(X)$, $F_7(X)$.\\
Consider finally the field extension
${\mathbb{F}}_{32}/{\mathbb{F}}_2$. Observe that $31$ is a
prime. Hence, all the polynomials $F_{i_s}$, $i_s >0$,  have the property of
Proposition~\ref{proptracelike}. These are
$F_1(X)=X^{16}+X^8+X^4+X^2+X$, $F_3(X)=X^{24}+X^{17}+X^{12}+X^6+X^3$,
$F_5(X)=X^{20}+X^{18}+X^{10}+X^9+X^5$,
$F_7(X)=X^{28}+X^{25}+X^{19}+X^{14}+X^7$,
  $F_{11}(X)=X^{26}+X^{22}+X^{21}+X^{13}+X^{11}$, and $F_{15}(X)=X^{30}+X^{29}+X^{27}+X^{23}+X^{15}$.
\end{example}
\section{Codes from optimal generalised $C_{ab}$  polynomials}\label{secfour}
In
this section we consider codes from generalised $C_{ab}$ polynomials over ${\mathbb{F}}_q$ with $n=aq$
zeros. These polynomials are optimal in the sense that a bivariate polynomial with leading monomial $X^a$ can have no more
zeros over ${\mathbb{F}}_q$, as  is seen from the footprint bound Corollary~\ref{thefoot}.  Hence, we shall call
them {\em{optimal generalised $C_{ab}$ polynomials}}. We list a couple of
properties of optimal generalised $C_{ab}$ polynomials $F(X,Y)=X^a+\alpha Y^b+R(X,Y)$. It holds that $a<b$ and that 
$\{F(X,Y), Y^{q}-Y\}$ 
constitutes
a Gr\"{o}bner basis ${\mathcal{G}}$ for $I_q=\langle F(X,Y), X^q-X, Y^q-Y\rangle$ with respect to $\prec_w$. Here, and in the remaining part of the section, $\prec_w$ is the weighted degree lexicographic ordering in Definition~\ref{defwdeg} with weights as in Definition~\ref{defpab} and with $X=X_1$, $Y=X_2$. Furthermore, 
$\{M_1, \ldots , M_n\}=\Delta_{\prec_w}(I_{q})=\{ X^{i_1}Y^{i_2} \mid 0
\leq i_1 < a, 0 \leq i_2 < q\}$. Recall, that we assume $M_1 \prec_w \cdots \prec_w M_n$.\\
From the previous section we
have a simple method for constructing optimal generalised $C_{ab}$ polynomials over ${\mathbb{F}}_q={\mathbb{F}}_{p^m}$, where $p$ is a prime power and $m$ is an integer greater or equal to $2$. The method consists in letting
$F(X,Y)=G(X)-H(Y)$ where $G(X)$ is the trace polynomial and $H(Y)$ is
an arbitrary non-trivial
$({\mathbb{F}}_{p^m},{\mathbb{F}}_p)$-polynomial. We stress that the
results of the present section hold for any optimal generalised
$C_{ab}$ polynomial over arbitrary finite field ${\mathbb{F}}_q$.
The main result of the section is:
\begin{theorem} \label{teoalpha}
Let $I_{q}$ be defined from an optimal generalised $C_{ab}$ polynomial and
let the weights $w(X)$ and $w(Y)$ be as in Definition~\ref{defpab}. 
Consider $\vec{c} = \mathrm{ev}(\sum_{s=1}^i a_s M_s +
I_q),$ $a_s \in \bbF_q$, $s=1,\ldots,i$ and
$a_i \neq 0$. Write $M_i = X^{\alpha_1} Y^{\alpha_2}$ and
$T=\alpha_1 {\mbox{ rem }} w(Y)$.
We have that $$w_H(\vec{c}) \geq (a-\alpha_1)(q-\alpha_2) + \epsilon \mbox{ where}$$
$$\epsilon = \begin{cases}
0 				& \mbox{if }q-b \leq \alpha_2 < q \\
T(q-\alpha_2-b) 		& \mbox{if }0  \leq \alpha_1 \leq a-w(Y)\\
				& \mbox{and }0 \leq \alpha_2 < q-b \\
\alpha_1(q-\alpha_2-b) 		& \mbox{if }a-w(Y) < \alpha_1 < a \mbox{ and }\\
				& q-w(X)-\alpha_1\frac{b-w(X)}{a-w(Y)} < \alpha_2 < q-b \\ 
T(q-\alpha_2-w(X)) 		& \mbox{if }a-w(Y) < \alpha_1 < a \mbox{ and }\\
				& 0 \leq \alpha_2 \leq q-w(X)-\alpha_1\frac{b-w(X)}{a-w(Y)}. \end{cases}$$
\end{theorem}
The proof of Theorem~\ref{teoalpha} calls for a definition
and some lemmas. Recall from Theorem~\ref{thenewbound} that we need to
estimate the size of the sets ${\mathcal{L}}(u)$, $u=1, \ldots ,
v+1$. For this purpose we introduce the following related sets:
\begin{definition}\label{defiB}
Let the notation be as in Definition~\ref{defpab} and Theorem~\ref{teoalpha}. For arbitrary $\alpha_1, \alpha_2$, $0 \leq
\alpha_1 <a$, $0\leq \alpha_2 <q$ we define
$$
\begin{array}{l}
B_1(X^{\alpha_1}Y^{\alpha_2}) = \{X^{\gamma_1} Y^{\gamma_2}\mid \alpha_1 \leq \gamma_1 < a, \alpha_2 \leq \gamma_2 < q\}, \\
\ \\
B_2(X^{\alpha_1}Y^{\alpha_2}) = \\
\ \\
{\mbox{ \ \ \ }} \left\{ \begin{array}{ll} 
				\bigg\{X^{\gamma_1}Y^{\gamma_2}\mid \alpha_1-T \leq \gamma_1 < \alpha_1,	\\
				{\mbox{ \ \hspace{4cm} }} \alpha_2+b \leq \gamma_2 < q \bigg\}        		& 	\begin{array}{l} {\mbox{\ if \ }} T \neq 0 \\
														                  	 {\mbox{\ and \ }} 0 \leq \alpha_2 < q-b
															\end{array}\\
				\ \\
				\emptyset 									&	{\mbox{\hspace{1.3mm} otherwise, \ }}
		\end{array} \right.\\
\ \\
\mbox{and for $u =1,\ldots,\gcd(a,b)$}\\
\ \\
B_3(X^{\alpha_1}Y^{\alpha_2},u) = \\
\ \\
{\mbox{ \ \ \ }} \left\{ \begin{array}{ll}
				\bigg\{X^{\gamma_1}Y^{\gamma_2}\mid a-w(Y)u \leq \gamma_1 < \alpha_1,\\
				{\mbox{ \ \hspace{0.7cm}}} \alpha_2+w(X)u \leq \gamma_2 < q\bigg\} 	& \begin{array}{l} 	{\mbox{ \ if \ }} a-w(Y) < \alpha_1 < a  \\
																{\mbox{ \ and \ }} 0 \leq \alpha_2 < q-b
													\end{array} \\
				\ \\
				\emptyset 								&			{\mbox{\ {\hspace{2.3mm}} otherwise.}}
		\end{array} \right.\\
\end{array}
$$
\end{definition}
\begin{remark} \label{remarkB}
Note that $w(X)\gcd(a,b) = b$ and $w(Y)\gcd(a,b) = a$, thus: $$B_3(X^{\alpha_1}Y^{\alpha_2},\gcd(a,b)) = \bigg\{X^{\gamma_1}Y^{\gamma_2} \mid 0 \leq \gamma_1 < \alpha_1,\alpha_2+b \leq \gamma_2 < q\bigg\}.$$
Furthermore for any choice of $u \in \{1,\ldots,\gcd(a,b)\}$ and $M \in \Delta_\prec(I_q)$ we have that $B_1(M) \cap B_2(M) = B_1(M) \cap B_3(M,u) = \emptyset$. If $B_3(M,u) \neq \emptyset$ then $B_2(M)\subseteq B_3(M,u)$.
\end{remark}
Before continuing with the lemmas we illustrate Definition \ref{defiB} with an example.
\begin{example} \label{exB}
Consider an optimal generalised $C_{ab}$ polynomial $F(X,Y)=X^9-Y^{12}+R(X,Y) \in {\mathbb{F}}_{27}[X,Y]$. We have $a=9$, $b=12$, $w(X)=4$, $w(Y)=3$, and $\Delta_{\prec_w}(I_q)=\{X^{i_1}Y^{i_2} \mid 0 \leq i_1 < 9, 0 \leq i_2 < 27\}$.\\
We first treat the case $X^{\alpha_1}Y^{\alpha_2}=X^5Y^{16}$. We have $\alpha_2 \geq q-b$, thus $B_2(X^{\alpha_1}Y^{\alpha_2})=B_3(X^{\alpha_1}Y^{\alpha_2},u)=\emptyset$ for any $u$. For an illustration see Figure \ref{figure1a}.\\
Now consider the case $X^{\alpha_1}Y^{\alpha_2}=X^5Y^{4}$. We have
$\alpha_2 < q-b$ and $T = 2 \neq 0$ and therefore
$B_2(X^{\alpha_1}Y^{\alpha_2})$ is non-empty. Because $T=2$, the width
of $B_2(X^{\alpha_1}Y^{\alpha_2})$ is $2$. Turning to
$B_3(X^{\alpha_1}Y^{\alpha_2},u)$ we see that $\alpha_1<a-w(Y)$ and
therefore the sets $B_3(X^{\alpha_1}Y^{\alpha_2},u)$'s are empty. See Figure \ref{figure1a} for an illustration.\\
\begin{figure}
\begin{center}
	\includegraphics[width=0.90\textwidth]{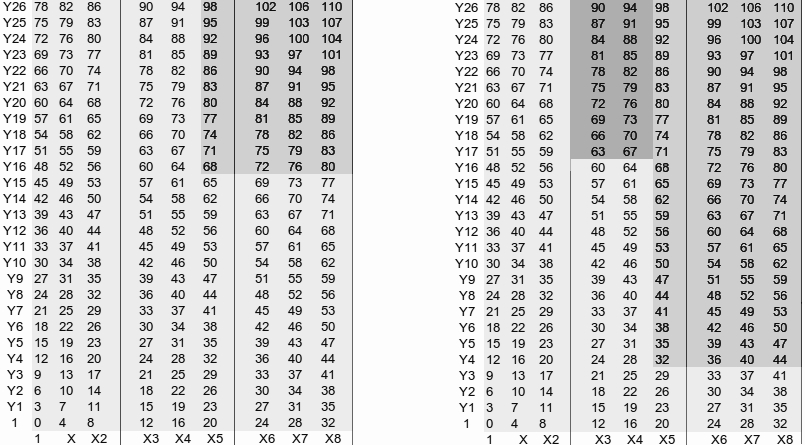}
\end{center}
\caption{Left part:   $X^{\alpha_1}Y^{\alpha_2}=X^5Y^{16}$. Only $B_1$ present. 
Right part: $X^{\alpha_1}Y^{\alpha_2}=X^5Y^{4}$. Light grey area is
$B_1$, medium grey area is  $B_2$. $B_3$ is not present.}
\label{figure1a}
\end{figure}
\noindent Consider next the case
$X^{\alpha_1}Y^{\alpha_2}=X^8Y^{3}$. We have  $\alpha_2 < q-b$ and
$\alpha_1 > a-w(Y)$ and therefore $B_2(X^{\alpha_1}Y^{\alpha_2})$ and
$B_3(X^{\alpha_1}Y^{\alpha_2},u)$ for $u=1,2,3$ are non-empty. The
situation regarding $B_2(X^{\alpha_1}Y^{\alpha_2})$ is similar to the
case $X^5Y^4$. The set $B_3(X^{\alpha_1}Y^{\alpha_2},u)$ can be
thought of as an improvement to $B_2(X^{\alpha_1}Y^{\alpha_2})$. We
see that $\gamma_1$ runs from $a-w(Y)u$ to $\alpha_1$ and $\gamma_2$ from $\alpha_2+w(X)u$ to $q$.
For an illustration see Figure \ref{figure2a}.\\
\begin{figure}[!h]
\begin{center}
	\includegraphics[width=0.90\textwidth]{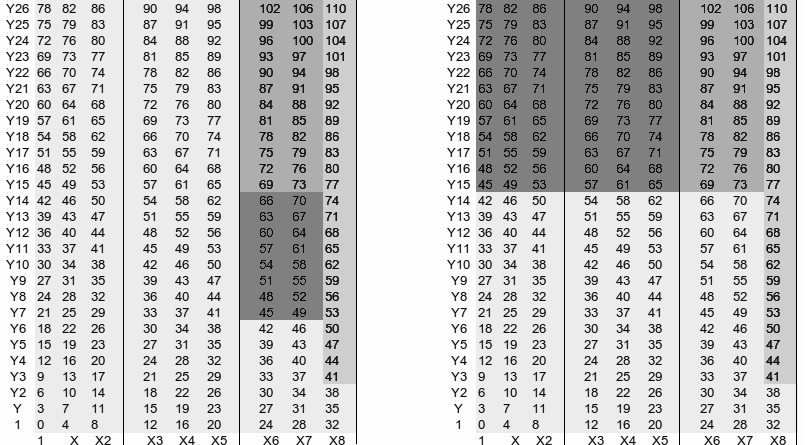}
\end{center}
\caption{In both parts $X^{\alpha_1}Y^{\alpha_2}=X^8Y^{3}$.
Left part: Light grey area is $B_1$, medium grey area is $B_2$, and
dark grey area plus medium grey area correspond to $B_3(X^{\alpha_1}Y^{\alpha_2},1)$. 
Right part: Light grey area is $B_1$, medium grey area is $B_2$, and
dark grey area plus medium grey area correspond to $B_3(X^{\alpha_1}Y^{\alpha_2},3)$.} 
\label{figure2a}
\end{figure}
\end{example}
\begin{lemma} \label{propBset}
Consider $\vec{c}={\mbox{ev}}(\sum_{s=1}^i a_s M_s +I_q)$,
$a_s \in {\mathbb{F}}_q$, $s=1, \ldots , i$, and $a_i \neq 0$. 
Let $M_i = X^{\alpha_1} Y^{\alpha_2}$ and $v= \alpha_1 {\mbox{ div }} w(Y)$ (that is, $v$ satisfies $\alpha_1 = w(Y)v+T$, where $T = \alpha_1 {\mbox{ rem }} w(Y)$). It holds that:
\begin{itemize}
\item $B_1(X^{\alpha_1}Y^{\alpha_2}) \subseteq {\mathcal{L}}(u) \mbox{ for }u =1,\ldots,v+1$.
\item $B_2(X^{\alpha_1}Y^{\alpha_2}) \subseteq {\mathcal{L}}(u) \mbox{ for }u =1,\ldots,v+1$.
\item $B_3(X^{\alpha_1}Y^{\alpha_2},\gcd(a,b)) \subseteq {\mathcal{L}}(v+1)$.
\item $B_3(X^{\alpha_1}Y^{\alpha_2},u) \subseteq {\mathcal{L}}(u) \mbox{ for }u =1,\ldots,v$.
\end{itemize}
\end{lemma}
\begin{proof} \ \\
\noindent {\underline{$B_1(X^{\alpha_1}Y^{\alpha_2}) \subseteq
    {\mathcal{L}}(u)$ for $u =1, \ldots,v+1$:}}\\
Assume $M_l = X^{\gamma_1}Y^{\gamma_2} \in
B_1(X^{\alpha_1}Y^{\alpha_2})$. We have $\alpha_1 \leq \gamma_1 < a$
and $\alpha_2 \leq \gamma_2 < q$. Choosing $M_j =
X^{\gamma_1-\alpha_1}Y^{\gamma_2-\alpha_2}$ we get $\mbox{lm}(M_{i}M_j
{\mbox{ rem }} {\mathcal{G}}) = M_l$. Let $i' \in \{1,\ldots,i-1\}$, then by the
properties of a monomial ordering $M_{i'}M_j \prec_w M_iM_j$ holds. This
means that $(M_i,M_j)$ is SOWB with respect the set
$\{1,\ldots,i\}$. Thus $M_l \in {\mathcal{L}}(u)$ for $u =1,
\ldots,v+1$.\\

\noindent {\underline{$B_2(X^{\alpha_1}Y^{\alpha_2}) \subseteq {\mathcal{L}}(u)$ for $u =1, \ldots,v+1$:}}\\
If $T=0$ or $q-b \leq \alpha_2 < q$ then the result follows trivially.\\
Assume $T \neq 0$ and $0 \leq \alpha_2 < q-b$. Let $M_l = X^{\gamma_1}Y^{\gamma_2} \in B_2(X^{\alpha_1}Y^{\alpha_2})$. We have $\alpha_1-T \leq \gamma_1 < \alpha_1$ and $\alpha_2 + b \leq \gamma_2 < q$. Choosing $M_j = X^{\gamma_1-\alpha_1+a}Y^{\gamma_2-\alpha_2-b}$ (which belongs to $\Delta_{\prec_w}(I_q)$ by the definition of $B_2$) we get
$${\mbox{lm}}(M_{i}M_j {\mbox{\ rem\ }} {\mathcal{G}})={\mbox{lm}}(M_{i}M_j -X^{\gamma_1}Y^{\gamma_2-b}F(X,Y))=X^{\gamma_1}Y^{\gamma_2}.$$
We want to prove that $(M_i,M_j)$ is SOWB with respect the set
$\{1,\ldots,i\}$. We consider $M_{i'}$ with $i' \in
\{1,\ldots,i-1\}$. If $w(M_{i'}) < w(M_i)$ then the proof follows from $w(M_{i'}M_j) <  w(M_{i}M_j)$ using the fact that reducing modulo $F$ does not change the weight of the leading monomial. If $w(M_{i'}) = w(M_i)$ then there exists an integer $z$ with $\alpha_1 - zw(Y) \geq 0$ such that $M_{i'} = X^{\alpha_1 - zw(Y)}Y^{\alpha_2 + zw(Y)}$.  Therefore $\gamma_1 -zw(Y) \geq 0.$\\
Now $M_{i^\prime}M_j=X^{a+\gamma_1-zw(Y)}Y^{\gamma_2-b+zw(X)}$ and
therefore
$${\mbox{lm}}(M_{i^\prime}M_j {\mbox{ rem }} {\mathcal{G}}) = {\mbox{lm}}(M_{i^\prime}M_j - X^{\gamma_1-zw(Y)}Y^{\gamma_2-b+zw(X)}F(X,Y)) $$
$$= X^{\gamma_1-zw(Y)}Y^{\gamma_2+zw(X)} \prec_w X^{\gamma_1}Y^{\gamma_2}.$$
Again we employed the fact that reducing modulo $F$ does not change
the weight of the leading monomial. We conclude that
${\mbox{lm}}(M_{i^\prime}M_j {\mbox{\ rem\ }} {\mathcal{G}}) \prec_w
X^{\gamma_1}Y^{\gamma_2}$ and that $(M_i,M_j)$ is SOWB with respect
the set $\{1,\ldots,i\}$. Thus $M_l \in {\mathcal{L}}(u)$ for $u
=1, \ldots,v+1$.\\

\noindent {\underline{$B_3(X^{\alpha_1}Y^{\alpha_2},\gcd(a,b)) \subseteq {\mathcal{L}}(v+1)$:}}\\
If $0 \leq \alpha_1 \leq a-w(Y)$ or $q-b \leq \alpha_2 < q$ then the result follows trivially.\\
Assume $a-w(Y) < \alpha_1 < a$ and $0 \leq \alpha_2 < q-b$, then $v=\gcd (a,b)-1$. Let $M_l =
X^{\gamma_1}Y^{\gamma_2} \in
B_3(X^{\alpha_1}Y^{\alpha_2},\gcd(a,b))$. We have $0 \leq \gamma_1 <
\alpha_1$ and $\alpha_2+b \leq \gamma_2 < q$. Choosing $M_j =
X^{\gamma_1-\alpha_1+a}Y^{\gamma_2-\alpha_2-b}$ we get
$\mbox{lm}(M_iM_j \mbox{ rem } {\mathcal{G}}) = M_l$. We want to prove
that $(M_i,M_j)$ is SOWB with respect the set $\{1,\ldots,i-v-1\}$. We
consider $M_{i'}$ with $i' \in \{1,\ldots,i-1\}$. If $w(M_{i'}) <
w(M_i)$ the proof follows because $w(M_{i'}M_j) <
w(M_{i}M_j)$ using the fact that reducing modulo $F$ does not change
the weight of the leading monomial. As $v = \gcd(a,b)-1$ there 
does not exists any $i' \in \{1,\ldots,i-v-1,i\}$ such that $w(M_{i'}) = w(M_{i})$. From this it follows that $(M_i,M_j)$ is SOWB with respect the set $\{1,\ldots,i-v-1\}$ and thus $M_l \in {\mathcal{L}}(v+1)$.\\

\noindent {\underline{$B_3(X^{\alpha_1}Y^{\alpha_2},u) \subseteq {\mathcal{L}}(u)$ for $u =1, \ldots,v$:}}\\
If $q-b \leq \alpha_2 < q$ or $0 \leq \alpha_1 \leq a-w(Y)$ then the result follows trivially.\\
Assume $a-w(Y) < \alpha_1 < a$ and $0 \leq \alpha_2 < q-b$, then $v=\gcd (a,b)-1$. Let $M_l =
X^{\gamma_1}Y^{\gamma_2} \in B_3(X^{\alpha_1}Y^{\alpha_2},u)$. We have
$a-w(Y)u \leq \gamma_1 < \alpha_1$ and $\alpha_2+w(X)u \leq \gamma_2 <
q$. By the definition of $\prec_w$ and the form of
$\Delta_{\prec_w}(I_q)$  we have that
$M_{i-u}=X^{\alpha_1-w(Y)u}Y^{\alpha_2+w(X)u}$. Choosing $M_j =
X^{\gamma_1-\alpha_1+w(Y)u}Y^{\gamma_2-\alpha_2-w(Y)u}$ we get
$\mbox{lm}(M_{i-u}M_j {\mbox{ rem }} {\mathcal{G}}) = M_l$. Note that $M_{i-u}$
and $M_j$ are in $\Delta_{\prec_w}(I_q)$ because
$v=\gcd(a,b)-1$, $a-w(Y) < \alpha_1 < a$ and $0 \leq \alpha_2 <
q-b$. We want to prove that $(M_i,M_j)$ is SOWB with respect the set
$\{1,\ldots,i-u,i\}$. We consider $M_{i'}$ with $i' \in
\{1,\ldots,i-1\}$. If $w(M_{i'}) < w(M_i)$ then the proof follows from $w(M_{i'}M_j) <  w(M_{i}M_j)$ using the fact that reducing modulo $F$ does not change the weight of the leading monomial. The monomials $M_{i'}$ which satisfy $w(M_{i'}) =
w(M_{i-u})$ are $M_i$ and $M_{i-z}$ for $z=u,\ldots,v$. However, 
$M_iM_j {\mbox{ rem }} {\mathcal{G}} \prec_w M_{i-u}M_j {\mbox{ rem }} {\mathcal{G}}$
because $\gamma_1+w(Y)u > a$ and $M_{i-t}M_j \prec_w M_{i-u}M_j$ for
any $t=u+1,\ldots,v$ due to the properties of a monomial
ordering. From this it follows that $(M_i,M_j)$ is SOWB with respect
the set $\{1,\ldots,i-u,i\}$ and thus $M_l \in {\mathcal{L}}(u)$, for
$u=1, \ldots , v$.
\end{proof}
\begin{lemma} \label{corB}
Consider $\vec{c}={\mbox{ev}}(\sum_{s=1}^i a_s M_s +I_q)$,
$a_s \in {\mathbb{F}}_q$, $s=1, \ldots , i$, and $a_i \neq 0$. 
Write $M_i = X^{\alpha_1} Y^{\alpha_2}$. For $u =1,\ldots,v+1$, with $v=\alpha_1 {\mbox{ div }} w(Y)$, we have that:
$$B_1(X^{\alpha_1}Y^{\alpha_2}) \cup B_2(X^{\alpha_1}Y^{\alpha_2}) \subseteq {\mathcal{L}}(u),$$
$$B_1(X^{\alpha_1}Y^{\alpha_2}) \cup B_3(X^{\alpha_1}Y^{\alpha_2},u) \subseteq {\mathcal{L}}(u).$$
\end{lemma}
\begin{proof}
The lemma follows directly from Remark \ref{remarkB} and Lemma~\ref{propBset}.
\end{proof}
It is not hard  to compute the cardinality of the sets $B_1$, $B_2$
and $B_3$. For $u =1,\ldots,\gcd(a,b)$, we have that:
$$\# B_1(X^{\alpha_1}Y^{\alpha_2}) = (a-\alpha_1)(q-\alpha_2),$$
$$\# B_2(X^{\alpha_1}Y^{\alpha_2}) = \begin{cases} 
\alpha_1(q-\alpha_2-b) &\mbox{ if }0 \leq \alpha_2 < q-b \\
0 & \mbox{ otherwise,} \end{cases}$$
$$\# B_3(X^{\alpha_1}Y^{\alpha_2},u) = \begin{cases} 
(w(Y)u-a+\alpha_1)(q-\alpha_2-w(X)u) & \mbox{ if }0 \leq \alpha_2 < q-b \mbox{ and } \\
				     & 	 a-w(Y) < \alpha_1 < a \\
0 & \mbox{ otherwise.} \end{cases}$$

Thus, for $u =1,\ldots,v+1$ by Lemma \ref{corB} we get:
$$\#{\mathcal{L}}(u) \geq (a-\alpha_1)(q-\alpha_2) +  
\begin{cases} 
\alpha_1(q-\alpha_2-b) 	& \mbox{ if }0 \leq \alpha_2 < q-b \\
0 			& \mbox{ otherwise} 
\end{cases}$$
And if $a-w(Y) < \alpha_1 < a$:
$$\# {\mathcal{L}}(u) \geq (a-\alpha_1)(q-\alpha_2) +  
\begin{cases} 
(w(Y)u-a+\alpha_1)(q-\alpha_2-w(X)u) 	& \mbox{ if }0 \leq \alpha_2 < q-b \\
0 					& \mbox{ otherwise.} 
\end{cases}$$
Now we can prove Theorem \ref{teoalpha}.
\begin{proof}[Proof of Theorem \ref{teoalpha}]
Let $v = \alpha_1 {\mbox{ div }} w(Y)$.
If $0 \leq \alpha_1 \leq a-w(Y)$ then 
we obtain  
\begin{eqnarray*}
w_H(\vec{c}) 
& \geq & \min \{\# {\mathcal{L}}(1), \ldots , \#  {\mathcal{L}}(v+1)\} \\
	     & \geq & (a-\alpha_1)(q-\alpha_2) +\\
&&  				
\begin{cases} 
				\alpha_1(q-\alpha_2-b) 	& \mbox{ if }0 \leq \alpha_2 < q-b \\
				0 					& \mbox{ otherwise.} 
				\end{cases} \\
\end{eqnarray*}
If $a-w(Y) < \alpha_1 < a$, then $v=\gcd (a,b)-1$ and we obtain 
\begin{eqnarray*}
w_H(\vec{c}) & \geq & \min\{  \#
{\mathcal{L}}(1),\ldots , \#
{\mathcal{L}}(v+1)  \}\\
	     & \geq & (a-\alpha_1)(q-\alpha_2) +\\
&&  
				\begin{cases} 
				\min\{
                                (w(Y)u-a+\alpha_1)(q-\alpha_2-w(X)u)
                                \mid u=1,\ldots ,v+1 \}	& \mbox{ if }0 \leq \alpha_2 < q-b \\
				0 					& \mbox{ otherwise.} 
				\end{cases} \\
\end{eqnarray*}
The function $f(u)=(w(Y)u-a+\alpha_1)(q-\alpha_2-w(X)u)$ is a concave parabola, thus we have minimum in $u=1$ or $u=v+1=\gcd(a,b)$. 
By inspection $f(1)=(w(Y)-a+\alpha_1)(q-\alpha_2-w(X))=T(q-\alpha_2-w(X))$  and 
$f(\gcd(a,b))=(w(Y)\gcd(a,b)-a+\alpha_1)(q-\alpha_2-w(X)\gcd(a,b))=\alpha_1(q-\alpha_2-b)$.
We therefore get the biimplication:
\begin{eqnarray*}
&&f(1)  \leq  f(\gcd(a,b)), \\
&\Updownarrow \\
&&\alpha_2  \leq  q-w(X)-\alpha_1\frac{b-w(X)}{a-w(Y)},
\end{eqnarray*}
and the theorem follows.
\end{proof}
\begin{remark}
If for codes from optimal generalised $C_{ab}$ polynomials rather than applying
Theorem~\ref{thenewbound} we apply the usual Feng-Rao bound (Theorem~\ref{theseven}) then the $\epsilon$ in
Theorem~\ref{teoalpha} should be replaced with:
$$\begin{cases}
0 				& \mbox{if }q-b \leq \alpha_2 < q \\
T(q-\alpha_2-b) 		& \mbox{and }0 \leq \alpha_2 < q-b. \\
\end{cases}$$
We see that our new bound improves the  Feng-Rao bound by
$$\begin{cases}
0 				& \mbox{if }q-b \leq \alpha_2 < q \\		
				& \mbox{or }0  \leq \alpha_1 \leq a-w(Y)\\
(\alpha_1-T)(q-\alpha_2-b) 	& \mbox{if }a-w(Y) < \alpha_1 < a \mbox{ and }\\
				& q-w(X)-\alpha_1\frac{b-w(X)}{a-w(Y)} < \alpha_2 < q-b \\ 
T(b-w(X)) 			& \mbox{if }a-w(Y) < \alpha_1 < a \mbox{ and }\\
				& 0 \leq \alpha_2 \leq q-w(X)-\alpha_1\frac{b-w(X)}{a-w(Y)}. \end{cases}$$
\end{remark}
\begin{remark}
It is possible to show that Theorem~\ref{teoalpha} is the strongest possible
result one can derive from Theorem~\ref{thenewbound} regarding the minimum
distance of codes from optimal generalised $C_{ab}$ polynomials.
\end{remark}
In the following we apply Theorem~\ref{teoalpha} in a number of
cases where $F(X,Y)=G(X)-H(Y) \in {\mathbb{F}}_{p^m}[X,Y]$ with $G(X)$ being the trace polynomial
and $H(Y)$ being an $({\mathbb{F}}_{p^m},{\mathbb{F}}_p)$-polynomial
of another degree. Recall from the discussion at the beginning of the
section that these are optimal
generalised $C_{ab}$ polynomials. The strength of our new bound Theorem~\ref{thenewbound}
and Theorem~\ref{teoalpha} lies in the cases where $a$ and $b$ are not
relatively prime, as for $a$ and $b$ relatively prime it reduces to the
usual Feng-Rao bound for primary codes (see the last part of
Remark~\ref{remtksh2}). The well-known norm-trace
polynomial corresponds to choosing  $H(Y)$ to
be the norm polynomial. This gives $a=p^{m-1}$ and $b=(p^m-1)/(p-1)$
which are clearly relatively prime. The related codes, which are
called norm-trace codes, are 
one-point algebraic geometric codes. As a measure for how good is our
new code constructions it seems fair to compare the outcome of
Theorem~\ref{teoalpha} for the cases of $\gcd(a,b)>1$ with the
parameters of the one-point algebraic geometric codes from norm-trace
curves over the same alphabet. The two corresponding 
sets of ideals have the same footprint $\Delta_{\prec_w}(I_q)$
and consequently the corresponding codes are of the same length. We remind the reader that it was
shown in~\cite{geil2003codes} that the Feng-Rao bound gives the true
parameters of the norm-trace codes. 
\begin{example}\label{exqis8}
In this example we consider optimal generalised $C_{ab}$ polynomials
derived from $({\mathbb{F}}_8,{\mathbb{F}}_2)$-polynomials.  The trace
polynomial $G(X)$ is of degree $a=4$  and from Example~\ref{extksh5} we see that
besides the norm polynomial which is of degree $b=7$  we can choose $H(Y)$ as
$F_3(Y)=Y^6+Y^5+Y^3$ which is of degree $b=6$. The corresponding codes
are of length $n=32$ over the alphabet ${\mathbb{F}}_8$. In Figure~\ref{grafo4678} below
we compare the parameters of the related two sequences of improved codes
$\tilde{E}_{imp}(\delta)$ (Definition~\ref{defimpcode}). For few choices of $\delta$ the norm-trace
code is the best, but for many choices of $\delta$, from $(a,b)=(4,6)$ we get
better codes. We note that the latter sequence of codes contains
two non-trivial codes that has the best known parameters according to the linear code
bound at~\cite{tysker}, namely $[n,k,d]$ equal to $[32,2,28]$ and $[32,15,12]$.
\begin{figure}
\begin{center}
	\includegraphics[width=0.60\textwidth]{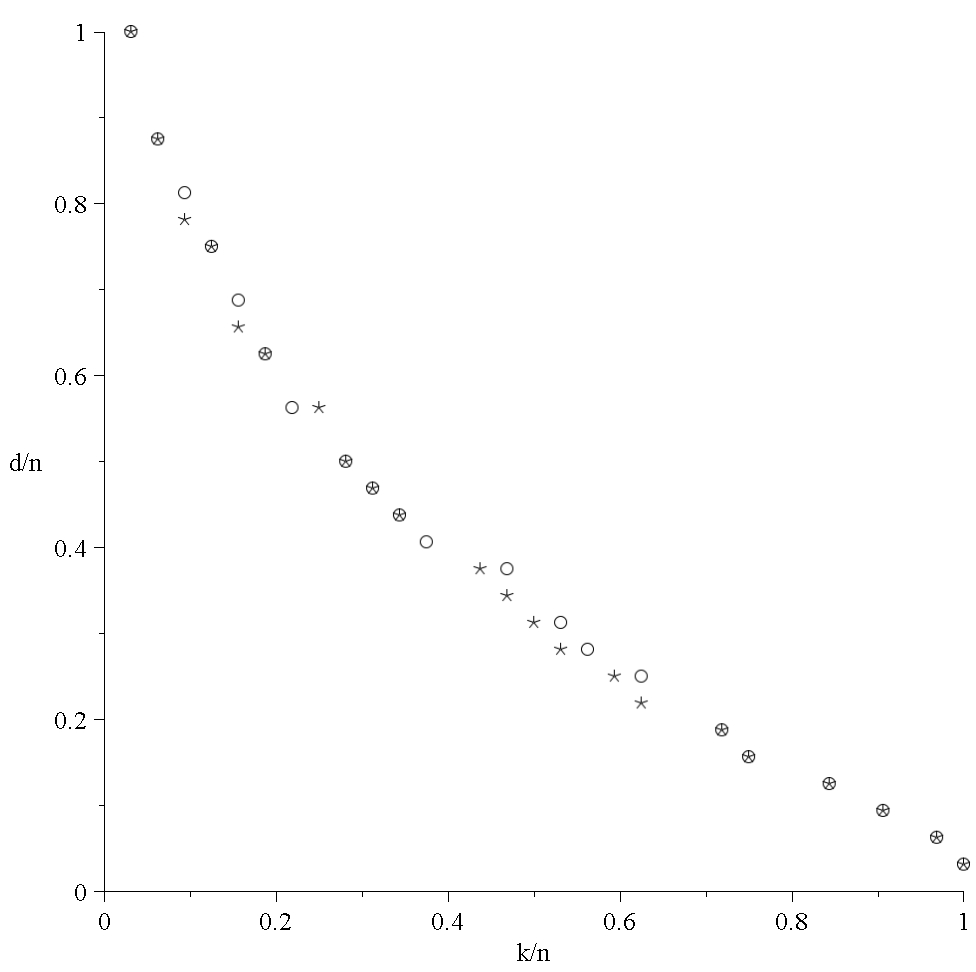}
\end{center}
\caption{Improved codes from Example~\ref{exqis8}. A $\circ$ corresponds
  to $(a,b)=(4,6)$, and an $\ast$ corresponds to $(a,b)=(4,7)$ (the norm-trace codes).}
\label{grafo4678}
\end{figure}
\end{example}
\begin{example}\label{exqis16}
In this example we consider optimal generalised $C_{ab}$ polynomials
derived from $({\mathbb{F}}_{16},{\mathbb{F}}_2)$-polynomials.  The trace
polynomial $G(X)$ is of degree $a=8$  and from Example~\ref{extksh5} we see that
besides the norm polynomial which is of degree $b=15$  we can choose
$H(Y)$ to be of degree $10$, $12$ and $14$. The corresponding codes
are of length $n=128$ over the alphabet ${\mathbb{F}}_{16}$. In Figure~\ref{grafo8101516} below
we compare the parameters of the related two sequences of improved codes
$\tilde{E}_{imp}(\delta)$ when $b=10$ and when $b=15$ (the norm-trace codes). For most choices of $\delta$ from $(a,b)=(8,10)$ we get
the best codes. The norm-trace codes are never strictly best.
\begin{figure}
\begin{center}
	\includegraphics[width=0.80\textwidth]{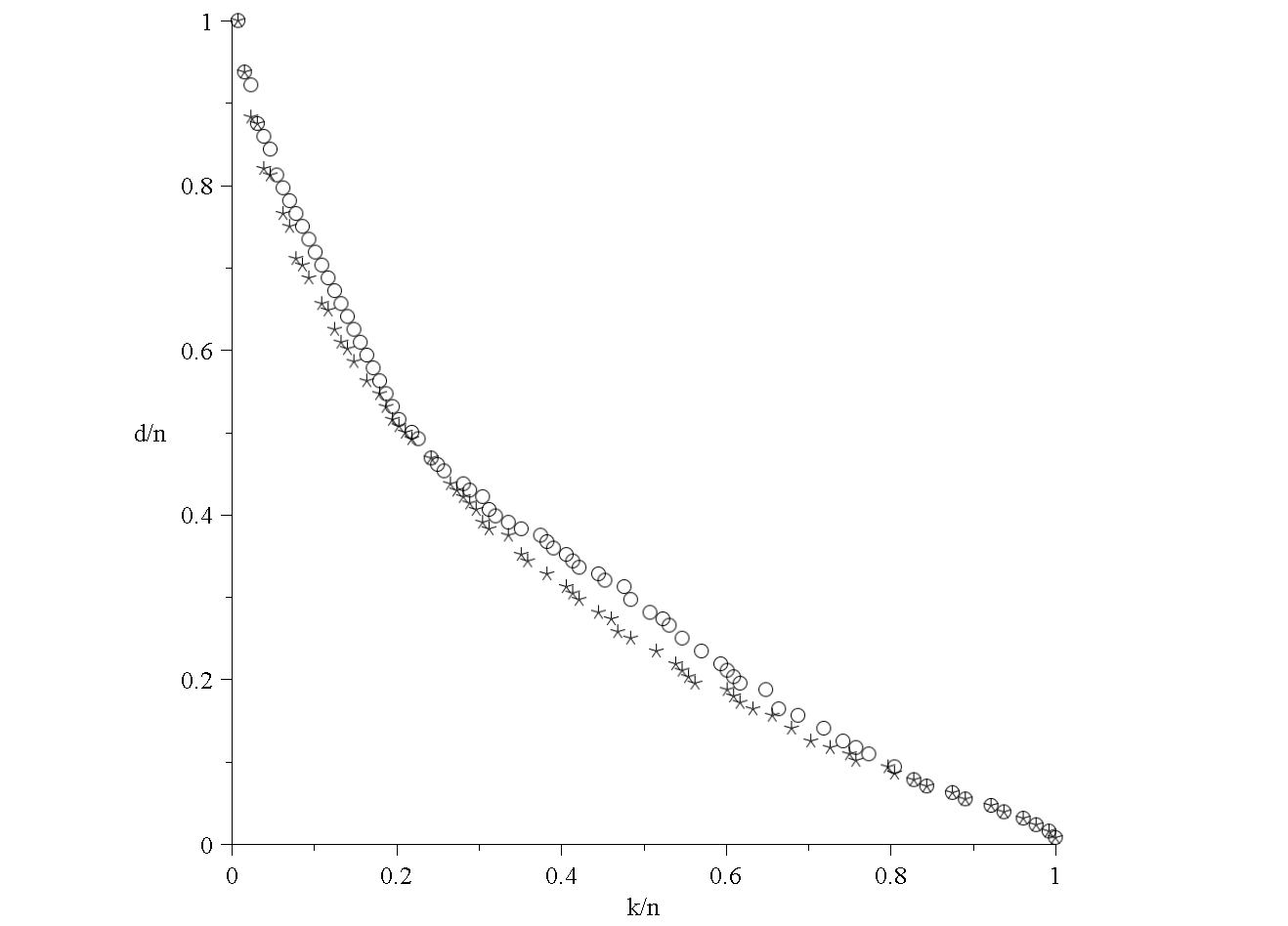}
\end{center}
\caption{Improved codes from Example~\ref{exqis16}. A $\circ$ corresponds
  to $(a,b)=(8,10)$, and an $\ast$ corresponds to $(a,b)=(8,15)$ (the
  norm-trace codes).}
\label{grafo8101516}
\end{figure}
\end{example}
\begin{example}\label{exqis32}
In this example we consider optimal generalised $C_{ab}$ polynomials
derived from $({\mathbb{F}}_{32},{\mathbb{F}}_2)$-polynomials.  The trace
polynomial $G(X)$ is of degree $a=16$  and from Example~\ref{extksh5} we see that
besides the norm-polynomial which is of degree $b=31$  we can choose
$H(Y)$ to be of degree $20$, $24$, $26$, $28$ and $30$. The corresponding codes
are of length $n=512$ over the alphabet ${\mathbb{F}}_{32}$. In
Figure~\ref{ZOOM1grafo1620263132}  
below
we compare the parameters of the related three sequences of improved codes
$\tilde{E}_{imp}(\delta)$ when $b=20$, $b=26$ and when $b=31$ (the
norm-trace codes). For no choices of $\delta$ the norm-trace codes are
strictly best (this holds for all values of $k/n$). For some choices $b=20$ gives the best codes for other
choices the best parameters are found by choosing $b=26$.
\begin{figure}
\begin{center}
\includegraphics[width=1.00\textwidth]{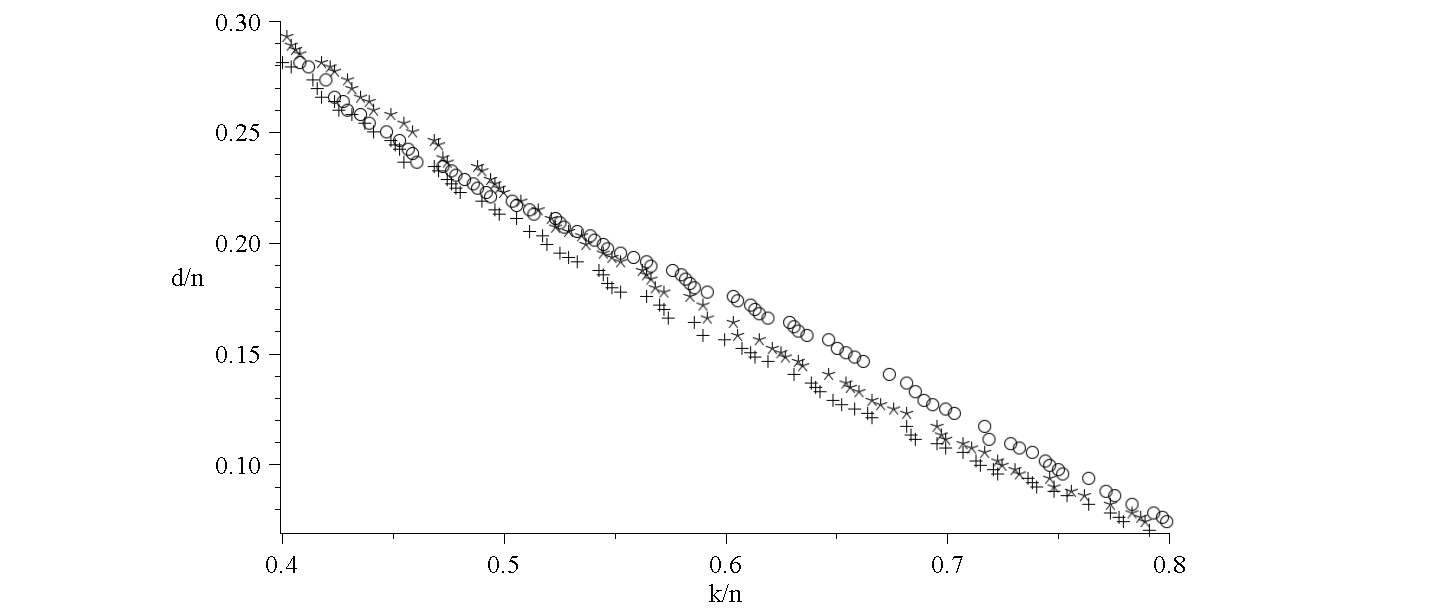}
\end{center}
\caption{Improved codes from Example~\ref{exqis32}. A $\circ$ corresponds
  to $(a,b)=(16,20)$, an $\ast$ to $(a,b)=(16,26)$, and finally a $+$  corresponds to $(a,b)=(16,31)$ (the
  norm-trace codes).}
\label{ZOOM1grafo1620263132}
\end{figure}
\end{example}
\begin{example}\label{exqis64}
In this example we consider optimal generalised $C_{ab}$ polynomials
derived from $({\mathbb{F}}_{64},{\mathbb{F}}_2)$-polynomials.  The trace
polynomial $G(X)$ is of degree $a=32$  and by studying cyclotomic
cosets we see that
as an alternative to the norm polynomial which is of degree $b=63$  we
can for instance choose an 
$H(Y)$ of degree $42$. The corresponding codes
are of length $n=2048$ over the alphabet ${\mathbb{F}}_{64}$. In Figure~\ref{grafo32426364} below
we compare the parameters of the related two sequences of improved codes
$\tilde{E}_{imp}(\delta)$ when $b=42$ and when $b=63$ (the norm-trace
codes). As is seen the first codes outperforms the last codes for
all parameters.
\begin{figure}
\begin{center}
	\includegraphics[width=0.80\textwidth]{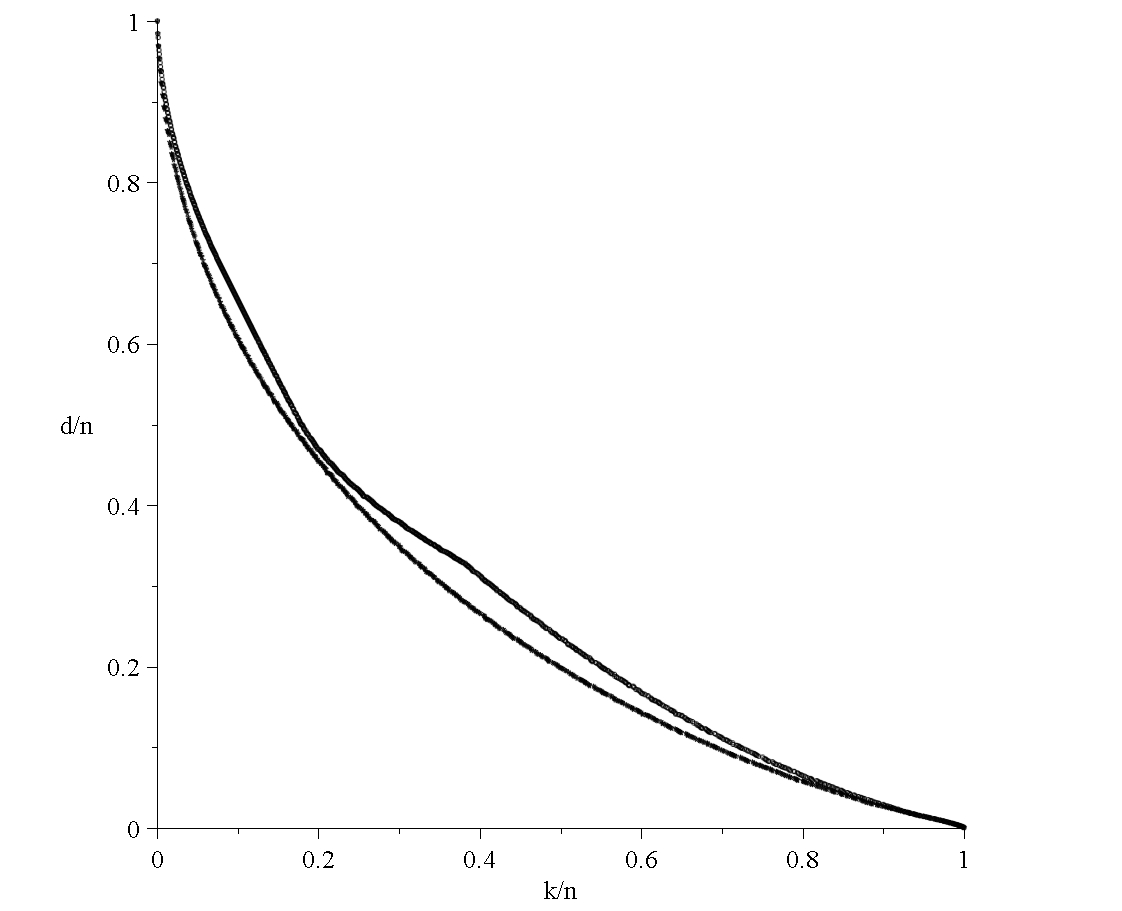}
\end{center}
\caption{Improved codes from Example~\ref{exqis64}. The upper curve corresponds
  to $(a,b)=(32,42)$, the lower curve to $(a,b)=(32,63)$ (the norm-trace codes)}
\label{grafo32426364}
\end{figure}
\end{example}
\section{A new construction of improved codes}\label{secfourandaquarter}
In Definition~\ref{defimpcode} we presented a Feng-Rao style
improved code construction $\tilde{E}_{imp}(\delta)$. As shall be
demonstrated in this section it is sometimes possible to do even
better. Recall that the idea behind Theorem~\ref{thenewbound} is to
consider case 1 up till case v+1 as described prior to the
theorem. Consider a general codeword
$$\vec{c} ={\mbox{ev}}\big(\sum_{s=1}^i a_s M_s +I_q\big) \in C(I,L)$$ 
$a_i \neq 0$, where $L$ is some fixed known subspace of ${\mathbb{F}}_q^n$. From $L$
we might {\it{a priori}} be able to conclude that certain $a_s$s
equal zero for all codewords as above. This corresponds to saying that
{\it{a priori}} we might know that some of the cases case 1 up to case
v do not happen. Clearly we could then leave out the corresponding
sets in Theorem~\ref{thenewbound}. This might result in a higher
estimate on $w_H(\vec{c})$. We illustrate the phenomenon with an
example in which we also show how to derive improved codes based on
this observation.
\begin{example}\label{exklein}
In this example we consider the Klein quartic $X^3Y+Y^3+X \in
{\mathbb{F}}_8[X,Y]$. Let $w(X)=2$ and $w(Y)=3$. The ideal $I=\langle X^3Y+Y^3+X\rangle \subseteq {\mathbb{F}}_8[X,Y]$ and the
corresponding weighted degree lexicographic ordering $\prec_w$ satisfy order domain
condition (C1) but not (C2) (as usual, in the definition of $\prec_w$ we choose $X=X_1$ and $Y=X_2$). Hence, it makes sense to apply
Theorem~\ref{thenewbound}. The footprint of  
$I_8=\langle X^3Y+Y^3+X,X^8+X,Y^8+Y\rangle$ 
is (for a reference
see~\cite[Ex.\ 4.19]{bookAG} and \cite[Ex.\ 3.3]{FR1}):
\begin{eqnarray*}
\Delta_{\prec_w}(I_8)&=&\{ 1, X, Y, X^2, XY, Y^2,X^3,X^2Y, XY^2,
X^4,Y^3,X^2Y^2,\\
&&X^5,XY^3,Y^4,X^6,X^2Y^3,XY^4,X^7,Y^5,X^2Y^4,Y^6\}
\end{eqnarray*}
written in increasing order with respect to $\prec_w$. Consider
$$\vec{c}={\mbox{ev}}\big( a_1
1+a_2X+a_3Y+a_4X^2+a_5XY+a_6Y^2+a_7X^3+I_8\big),$$
$a_7 \neq 0$. We have $w(X^3)=w(Y^2)>w(XY)$. Hence, by Remark~\ref{remtksh2} we
choose $v=1$. \\
By inspection the set corresponding to case 1 is
$${\mathcal{L}}(1)=\{X^3,X^4,X^5,X^6,X^7,X^2Y^4\}.$$
(Note that $X^2Y^4$ belongs to ${\mathcal{L}}(1)$ of the following
reason: We have ${\mbox{lm}}(X^3X^5 {\mbox{ rem }} X^8+X)=X$ and
${\mbox{lm}}(Y^2X^5 {\mbox{ rem }} X^3Y+Y^3+X)=X^2Y^4$, and from 
$w(Y^2X^5)=w(X^2Y^4)>w(X)$ we conclude that $(Y^2,X^5)$ is SOWB with
respect to $\{1,2,3,4,5,6,7\}$.) 
The set corresponding to case 2 is
$${\mathcal{L}}(2)=\{X^3,X^4,Y^3,X^5,XY^3,Y^4,X^6,X^2Y^3,XY^4,X^7,Y^5,X^2Y^4,Y^6\}.$$
If we know {\it{a priori}} that $a_6=0$ then  we can conclude
from the above that $w_H(\vec{c}) \geq \# {\mathcal{L}}(2)=13$. Without such an
information we can only conclude $$w_H(\vec{c}) \geq \min
\{\#{\mathcal{L}}(1),\#{\mathcal{L}}(2)\}=6.$$
It can be shown using Theorem~\ref{thenewbound} that
$\tilde{E}_{imp}(11)=C(I,L)$ where 
$$L={\mbox{ev}} \big( {\mbox{Span}}_{\mathbb{F}_8}\{1+I_8,
X+I_8,Y+I_8,X^2+I_8,XY+I_8,Y^2+I_8\}\big).$$
That is, a code with parameters $[n,k,d]$ equal to $[22,6, \geq
11]$.\\
If instead we choose
$$\tilde{L}={\mbox{ev}} \big( {\mbox{Span}}_{\mathbb{F}_8}\{1+I_8,
X+I_8,Y+I_8,X^2+I_8,XY+I_8,X^3+I_8\}\big)$$
then we do not need to consider the case 1 described above. By
inspection the code parameters $[n,k,d]$ of $C(I,\tilde{L})$ are $[22,
6, \geq 12]$.
\end{example}
\section{Generalised Hamming weights}\label{secfourandahalf}
As mentioned at the end of Section~\ref{sectwo} it is possible to lift
Theorem~\ref{thenewbound} to also deal with generalised Hamming
weights. Recall 
that these
parameters are important in the analysis of the wiretap channel of
type II as well as in the analysis of secret sharing schemes based on
coding theory, see~\cite{wei}, \cite{luo} and \cite{kurihara}.
\begin{definition}
Let $C \subseteq {\mathbb{F}}_q^n$ be a code of dimension $k$. For
$t=1, \ldots , k$ the $t$th generalised Hamming weight is
$$d_t(C)=\min \{ \# {\mbox{Supp}} \, D \mid D {\mbox{ \ is a subspace of
    \ }} C {\mbox{ \ of dimension \ }}
t\}.$$
Here, ${\mbox{Supp}} \, D$ means the entries for which some word in
$D$ is different from zero.
\end{definition}
Clearly, $d_1$ is nothing but the usual minimum
distance.
In Proposition~\ref{protemmeliglang} below we lift
Theorem~\ref{thenewbound} to deal with the second generalised Hamming
weight. From this the reader can understand how to treat any
generalised Hamming weight. 
\begin{proposition}\label{protemmeliglang}
Let $D \subseteq {\mathbb{F}}_q^n$ be a subspace of dimension
2. Write
$D={\mbox{Span}}_{\mathbb{F}_q} \{
{\mbox{ev}}(\sum_{s=0}^{i_1}a_s
M_s),{\mbox{ev}}(\sum_{s=0}^{i_2}b_s M_s)\}$. Here,
$\Delta_\prec(I_q)=\{M_1, \ldots , M_n\}$, $a_s \in
{\mathbb{F}}_q$, $b_s \in
{\mathbb{F}}_q$ with $a_{i_1} \neq 0$ and $b_{i_2} \neq
0$. Without loss of generality we may assume $i_1 \neq i_2$. Let $v_1$
and $v_2$ be integers satisfying $0\leq v_1 < i_1$ and $0 \leq v_2
<i_2$. 
We have
$$\# {\mbox{Supp}}(D) \geq \min \{ \# {\mathcal{L}}(z_1,z_2) \mid 1
\leq z_1 \leq v_1+1, 1 \leq z_2 \leq v_2+1 \}.$$ 
The above sets are defined as follows:
For $z=1, \ldots ,  v_1$ 
\begin{eqnarray}
&&{\mathcal{L}}(z,v_2+1)=\nonumber \\
&& \big\{ K \in \Delta_\prec(I_q) \mid \exists M_j \in \Delta_\prec(I_q)
{\mbox{ such that either }} (M_{i_1},M_j) {\mbox{ is SOWB}}
\nonumber \\
&& {\mbox{ with respect to }} \{1, \ldots ,
i_1-z,i_1\} {\mbox{ and }} {\mbox{lm}}(M_{i_1}M_j {\mbox{ rem }} {\mathcal{G}})=K 
\nonumber \\
&&{\mbox{ or }}
\nonumber \\
&&(M_{i_1-z},M_j) {\mbox{ is SOWB with respect to }} \{1, \ldots ,
i_1-z,i_1\}\nonumber \\
&&{\mbox{ and }} {\mbox{lm}}(M_{i_1-z}M_j {\mbox{ rem }} {\mathcal{G}})=K
\nonumber \\
&&{\mbox{ or }}
\nonumber \\
&&(M_{i_2},M_j) {\mbox{ is SOWB with respect to }} \{1,
\ldots , i_2-v_2-1\} \nonumber \\
&&{\mbox{ and }} {\mbox{lm}}(M_{i_2}M_j {\mbox{ rem }} {\mathcal{G}})=K\big\}.
\nonumber
\end{eqnarray}
For $z=1, \ldots , v_2$, ${\mathcal{L}}(v_1+1,z)$ is defined in a
similar way.\\ 
For $z_1=1, \ldots , v_1$ and $z_2=1, \ldots ,
v_2$ we have
\begin{eqnarray}
&&{\mathcal{L}}(z_1,z_2)=\nonumber \\
&& \big\{ K \in \Delta_\prec(I_q) \mid \exists M_j \in \Delta_\prec(I_q)
{\mbox{ such that for some $u \in \{1,2\}$}}\nonumber\\
&& (M_{i_u},M_j) {\mbox{ is SOWB  with respect to }} \{1, \ldots ,
i_u-z_u,i_u\} \nonumber \\
&&{\mbox{ and }} {\mbox{lm}}(M_{i_u}M_j {\mbox{ rem }} {\mathcal{G}})=K 
\nonumber \\
&&{\mbox{ or }}
\nonumber \\
&&(M_{i_u-z_u},M_j) {\mbox{ is SOWB with respect to }} \{1, \ldots ,
i_u-z_u,i_u\}\nonumber \\
&&{\mbox{ and }} {\mbox{lm}}(M_{i_u-z_u}M_j {\mbox{ rem }} {\mathcal{G}})=K\big\},
\nonumber 
\end{eqnarray}
and finally
\begin{eqnarray}
&&{\mathcal{L}}(v_1+1,v_2+1)=\nonumber \\
&& \big\{ K \in \Delta_\prec(I_q) \mid \exists M_j \in \Delta_\prec(I_q)
{\mbox{ such that }} (M_{i_u},M_j) {\mbox{ is SOWB}}
\nonumber \\
&& {\mbox{ with respect to }} \{1, \ldots ,
i_u-v_u-1\} {\mbox{ and }} {\mbox{lm}}(M_{i_u}M_j {\mbox{ rem }}
{\mathcal{G}})=K \nonumber \\
&& {\mbox{ for some }} u \in \{1,2\}\big\}.
\nonumber 
\end{eqnarray}
The second generalised Hamming weight of $C(I,L)$ is found by
repeating the above process for all possible choices of $i_1<i_2$
corresponding to the cases that $D \subseteq C(I,L)$. 
\end{proposition}
\begin{proof}
The proof is a straight forward enhancement of the proof for Theorem~\ref{thenewbound}.

\end{proof}
For the choice of $v_1$ and $v_2$ in Proposition~\ref{protemmeliglang} we refer to Remark~\ref{remtksh2}. Admittedly, the  proposition is rather technical. Nevertheless 
even its generalisation
to higher generalised Hamming weights can often be quite
manageable. We shall comment further on this in~Section~\ref{seccomp}. 

\section{Formulation at linear code level}\label{secfive}
As mentioned in the introduction the Feng-Rao bound for primary codes in its most
general form is a bound on any linear code described by means of a
generator matrix. All other versions of the bound, such as the order
bound for primary codes and the Feng-Rao bound for primary affine
variety codes, can be viewed as
corollaries to it. Below we reformulate the new bound in
Theorem~\ref{thenewbound} at the linear code level.\\
Let $n$ be a positive integer and $q$ a prime power. Consider a fixed ordered triple
$({\mathcal{U}},{\mathcal{V}},{\mathcal{W}})$ where
${\mathcal{U}}=\{\vec{u}_1, \ldots ,
\vec{u}_n\}$, ${\mathcal{V}}=\{\vec{v}_1, \ldots , \vec{v}_n\}$, and
${\mathcal{W}}=\{\vec{w}_1, \ldots , \vec{w}_n\}$ are three (possibly
different)  bases for
${\mathbb{F}}_q^n$ as a vector space over ${\mathbb{F}}_q$. We shall
always denote by ${\mathcal{I}}$ the set $\{1, \ldots , n\}$.
\begin{definition}\label{def1}
Consider a basis ${\mathcal{A}}=\{\vec{a}_1, \ldots , \vec{a}_n\}$ for
${\mathbb{F}}_q^n$ as a vector space over ${\mathbb{F}}_q$. We define a
function $\bar{\rho}_{\mathcal{A}}: {\mathbb{F}}_q^n \rightarrow \{0, 1, \ldots ,
n\}$ as follows. For $\vec{c} \in {\mathbb{F}}_q \backslash \{
\vec{0} \}$ we let $\bar{\rho}_{\mathcal{A}}(\vec{c})=i$ if $\vec{c} \in
{\mbox{Span}}_{\mathbb{F}_q}\{\vec{a}_1, \ldots , \vec{a}_i\}  \backslash
{\mbox{Span}}_{\mathbb{F}_q} \{\vec{a}_1, \ldots , \vec{a}_{i-1}\}$. Here, we used
the notion ${\mbox{Span}}_{\mathbb{F}_q}\,  \emptyset = \{\vec{0}\}$. Finally, we let 
$\bar{\rho}_{\mathcal{A}}(\vec{0})=0$.
\end{definition}
\noindent The component wise product plays a crucial role in the
linear code enhancement of Theorem~\ref{thenewbound}.
\begin{definition}\label{deftksh6}
The component wise product of two vectors $\vec{u}$ and $\vec{v}$ in
${\mathbb{F}}_q^n$ is defined by $(u_1, \ldots , u_n)\ast (v_1,
\ldots, v_n)=(u_1v_1, \ldots , u_nv_n)$.
\end{definition}
\begin{definition}\label{def2}
Let $({\mathcal{U}},{\mathcal{V}},{\mathcal{W}})$ and ${\mathcal{I}}$ be as
above. 
Consider ${\mathcal{I}}^\prime
\subseteq {\mathcal{I}}$. An ordered pair $(i,j) \subseteq
{\mathcal{I}}^\prime \times {\mathcal{I}}$ is said to be one-way well-behaving (OWB) with
respect to ${\mathcal{I}}^\prime$ if $\bar{\rho}_{\mathcal{W}}(\vec{u}_{i^\prime} \ast
\vec{v}_j) < \bar{\rho}_{\mathcal{W}}(\vec{u}_{i} \ast
\vec{v}_j)$ holds for all $i^\prime \in {\mathcal{I}}^\prime$ with $i^\prime <
i$.
\end{definition}
The following theorem is a first
generalisation of the Feng-Rao bound for primary codes. The
generalisation compared to the usual Feng-Rao bound \cite{AG,geithom}
is that we allow ${\mathcal{I}}^\prime$
to be different from $\{1, \ldots
,\#  {\mathcal{I}}^\prime\}$. This is in the spirit of
Section~\ref{secfourandaquarter}.
\begin{theorem}\label{the2}
Consider $\vec{c}=\sum_{s=1}^t a_s \vec{u}_{i_s}$ with $a_s
\in {\mathbb{F}}_q$, $s=1, \ldots , t$, $a_t
\neq 0$ and $i_1< \cdots <i_t$. We have
\begin{equation}
\begin{array}{r} w_H(\vec{c}) \geq  
 \# \big\{ l\in {\mathcal{I}} \mid  \exists j \in {\mathcal{I}} {\mbox{ \  such that \
    }} \bar{\rho}_{\mathcal{W}}(\vec{u}_{i_t}\ast \vec{v}_j)=l,
  {\mbox{ \ \hspace{1cm} \ }} \\
(i_t,j) {\mbox{ \ is OWB with respect to \ }} \{ i_1,
\ldots , i_t\} \big\}.
\end{array} \label{eqth2}
\end{equation}
\end{theorem}
\begin{proof}
Let $l_1 < \cdots < l_\sigma$ be the indexes $l$ counted
in~(\ref{eqth2}). Denote by $j_1,\ldots , j_\sigma$ the corresponding
$j$-values from~(\ref{eqth2}). By assumption $\vec{c} \ast
\vec{v}_{j_1}, \ldots , \vec{c} \ast \vec{v}_{j_\sigma}$ are linearly
independent and therefore
$${\mbox{Span}}_{\mathbb{F}_q}
 \{\vec{c} \ast
\vec{v}_{j_1}, \ldots , \vec{c} \ast \vec{v}_{j_\sigma}\}= \vec{c}
\ast {\mbox{Span}}_{\mathbb{F}_q}
\{ \vec{v}_{j_1}, \ldots , \vec{v}_{j_\sigma}\}$$
is a vector space of dimension $\sigma$. The theorem follows from the fact that $\vec{c} \ast {\mathbb{F}}_q^n$
is a vector space of dimension $w_H(\vec{c})$ containing the above space. 
\end{proof}
A slight modification of Definition~\ref{def2} and the above proof
allows for further improvements.
\begin{definition}\label{deftksh7}
 Let ${\mathcal{I}}^\prime \subseteq {\mathcal{I}}$.
A pair $(i,j) \in {\mathcal{I}}^\prime \times {\mathcal{I}}$ is called strongly one-way well-behaving (SOWB)
with respect to ${\mathcal{I}}^\prime$ if
$\bar{\rho}_{\mathcal{W}}(\vec{u}_{i^\prime}\ast\vec{v}_j)<\bar{\rho}_{\mathcal{W}}(\vec{u}_i\ast
  \vec{v}_j)$ holds for all $i^\prime \in {\mathcal{I}}^\prime \backslash \{ i
  \}$. 
\end{definition}

The following theorem is the linear code interpretation of
Theorem~\ref{thenewbound}. Besides working for a larger class of codes, it is slightly stronger in that we formulate it in
such a way that it supports the technique explained in
Section~\ref{secfourandaquarter}. Concretely, what makes it stronger than Theorem~\ref{thenewbound}
is the presence of the set $\hat{\mathcal{I}}$.
\begin{theorem}\label{the4}
Consider a non-zero codeword $\vec{c}=\sum_{t=1}^i a_t \vec{u}_t$,
$a_t \in {\mathbb{F}}_q$ for $t=1,\ldots , i$, $a_i \neq 0$. Let $v$
be an integer $0 \leq v <i$. Assume that for some set
$\hat{\mathcal{I}}\subseteq \{1, \ldots , i-1\}$ we know {\it{a
    priori}} that $a_x=0$ when $x \in \hat{\mathcal{I}}$. Let $z_1<
\cdots <z_s$ be the numbers in $\{z \in \{i-v, \ldots , i-1\} \mid z \notin
\hat{\mathcal{I}}\}$. Write ${\mathcal{I}}^\ast=\{ z \in \{1, \ldots ,
i-v-1\}\mid z \notin \hat{\mathcal{I}}\}$. We have 
$$w_H(\vec{c})\geq \min \{{\mathcal{L}}^\prime(1), \ldots ,
{\mathcal{L}}^\prime(s+1) \}$$
where for $t=1, \ldots ,s$ we define ${\mathcal{L}}^\prime(t)$ as follows:
\begin{eqnarray}
{\mathcal{L}}^\prime(1)&=&\{ l \in {\mathcal{I}} \mid \exists z \in \{ z_s, i\} {\mbox{
    \ and \ }} j
\in {\mathcal{I}} {\mbox{ \ such that \ }} \bar{\rho}_{\mathcal{W}}(\vec{u}_z \ast
\vec{v_j})=l\nonumber \\
&&(z,j) {\mbox{ \ is SOWB with respect to
    \ }} {\mathcal{I}}^\ast \cup \{z_1, \ldots , z_s,i\} \big\}, \nonumber \\ 
{\mathcal{L}}^\prime(2)&=&\big\{ l \in {\mathcal{I}} \mid \exists z \in \{z_{s-1},i\}
{\mbox{ \ and \ }}  j \in {\mathcal{I}}  {\mbox{ \ such that \ }}
\bar{\rho}_{\mathcal{W}}(\vec{u}_{z} \ast \vec{v}_j) =l\nonumber \\
&&(z,j) {\mbox{ \ is SOWB
    with respect to \ }} {\mathcal{I}}^\ast \cup \{z_1, \ldots , z_{s-1},i\}\big\},\nonumber \\
&&{\mbox{ \ {\hspace{3cm} }}}\vdots \nonumber \\
{\mathcal{L}}^\prime(s)&=&\big\{ l \in {\mathcal{I}} \mid \exists z \in \{z_{1},i\}
{\mbox{ \ and \ }}  j \in {\mathcal{I}}  {\mbox{ \ such that \ }}
\bar{\rho}_{\mathcal{W}}(\vec{u}_{z} \ast \vec{v}_j) =l\nonumber \\
&& (z,j) {\mbox{ \ is SOWB
    with respect to \ }} {\mathcal{I}}^\ast \cup
\{z_1,i\}\big\}.\nonumber
\end{eqnarray}
Finally, 
\begin{eqnarray}
{\mathcal{L}}^\prime(s+1)&=&\big\{ l \in {\mathcal{I}} \mid \exists j \in {\mathcal{I}}  {\mbox{ \ such that \ }}
\bar{\rho}_{\mathcal{W}}(\vec{u}_{i} \ast \vec{v}_j) =l \nonumber \\
&&(i,j) {\mbox{ \ is OWB
    with respect to \ }} {\mathcal{I}}^\ast \cup \{i\}\big\}.\nonumber 
\end{eqnarray}
To establish a lower bound on the minimum distance of a code $C$ we
repeat the above process for each $i \in
\bar{\rho}_{\mathcal{U}}(C)$. For each such $i$ we choose a
corresponding $v$, defining an $s$, and we determine the sets
${\mathcal{L}}^\prime(1), \ldots , {\mathcal{L}}^\prime(s+1)$ and
calculate their cardinalities. The
smallest cardinality found when $i$ runs through $\bar{\rho}_{\mathcal{U}}(C)$ 
serves as a lower bound on the minimum distance.
\end{theorem}
\begin{proof}
The proof is a direct translation of the proof of Theorem~\ref{thenewbound}.
\end{proof}
\begin{remark}\label{remtksh8}
For $v=0$ Theorem~\ref{the4} reduces to Theorem~\ref{the2}. For higher values
of $v$ Theorem~\ref{the4} is at least as strong as Theorem~\ref{the2} and
sometimes stronger. In the same way as Theorem~\ref{thenewbound} was lifted in
Section~\ref{secfourandahalf} to deal with generalised Hamming weights one can lift
Theorem~\ref{the2} and Theorem~\ref{the4}.
\end{remark}
\section{A related bound for dual codes}\label{secsix}
In the recent paper~\cite{geilmartin2013further} we presented a new bound for dual codes.
This bound is an improvement to the Feng-Rao bound for
such codes as well as an improvement to the 
advisory bound from~\cite{salazar}. The new bound of the present paper
can be viewed as a natural counter part to the bound from~\cite{geilmartin2013further},
the one bound dealing with primary codes and the other with dual codes. 
\begin{definition}\label{def3}Consider an ordered triple of bases
$({\mathcal{U}},{\mathcal{V}},{\mathcal{W}})$ for ${\mathbb{F}}_q^n$ and ${\mathcal{I}}$ 
as in Section~\ref{secfive}.
We define $m : {\mathbb{F}}_q^n \backslash \{ \vec{0}\} \rightarrow
{\mathcal{I}}$ by $m(\vec{c})=l$ if $l$ is the smallest number in
${\mathcal{I}}$ for which 
$\vec{c} \cdot \vec{w}_l \neq 0$. (Here, and in the following the
symbol $\cdot$ means the usual inner product). 
\end{definition}
\begin{definition}\label{definvolved}
Consider numbers $1 \leq l, l+1, \ldots , l+g \leq n$. 
A set 
${\mathcal{I}}^\prime\subseteq {\mathcal{I}}$ is said to have the $\mu$-property with
respect to $l$ with exception $\{l+1, \ldots , l+g\}$ if for all $i
\in {\mathcal{I}}^\prime$ a $j\in {\mathcal{I}}$ exists such that
\begin{itemize}
\item[(1a)] $\bar{\rho}_{\mathcal{W}}(\vec{u}_i \ast \vec{v}_j)=l$, and
\item[(1b)] for
  all $i^\prime \in {\mathcal{I}}^\prime$ with $i^\prime < i$ either $\bar{\rho}_{\mathcal{W}}(\vec{u}_{i^\prime}\ast
  \vec{v}_j)< l$ or $\bar{\rho}_{\mathcal{W}}(\vec{u}_{i^\prime}\ast
  \vec{v}_j) \in \{l+1, \ldots , l+g\}$ holds.
\end{itemize}
Assume next that $l+g+1 \leq n$. The set ${\mathcal{I}}^\prime$ is
said to have the relaxed $\mu$-property with respect to $(l,l+g+1)$
with exception $\{l+1, \ldots , l+g\}$ if for all $i \in
{\mathcal{I}}^\prime$ a $j \in {\mathcal{I}}$ exists such that either
conditions $(1a)$ and $(1b)$ above hold or 
\begin{itemize}
\item[(2a)] $\bar{\rho}_{\mathcal{W}}(\vec{u}_i \ast \vec{v}_j)=l+g+1$,
  and
\item[(2b)] $(i,j)$ is OWB with respect to ${\mathcal{I}}^\prime$, and
\item[(2c)] no $i^\prime \in {\mathcal{I}}^\prime$ with $i^\prime < i$
  satisfies $\bar{\rho}_{\mathcal{W}}(\vec{u}_{i^\prime} \ast \vec{v}_j)=l$.
\end{itemize}
\end{definition}
The new bound from~\cite[Th.\ 19]{geilmartin2013further} reads:
\begin{theorem}\label{thenew}\label{thestrongone}
Consider a non-zero codeword $\vec{c}$ and let $l=m(\vec{c})$. Choose
a non-negative integer $v$ such that $l+v\leq n$. Assume that for some
indexes $x \in \{l+1, \ldots , l+v\}$ we know {\textit{a priori}} that
$\vec{c} \cdot \vec{w}_x=0$. Let $l^\prime_1< \cdots < l^\prime_s$ be
the remaining indexes from $\{l+1, \ldots , l+v\}$. Consider the sets
${\mathcal{I}}^\prime_0, {\mathcal{I}}_1^\prime, \ldots ,
{\mathcal{I}}_s^\prime$ such that:
\begin{itemize}
\item ${\mathcal{I}}^\prime_0$ has the $\mu$-property with respect to
  $l$ with exception $\{ l+1, \ldots , l+v\}$.
\item For $i=1, \ldots , s$, ${\mathcal{I}}^\prime_i$ has the relaxed
  $\mu$-property with respect to $(l, l^\prime_i)$ with exception
  $\{l+1, \ldots , l^\prime_i -1\}$.
\end{itemize}
We have 
\begin{equation}
w_H(\vec{c}) \geq \min \{ \# {\mathcal{I}}_0^\prime, \#
{\mathcal{I}}_1^\prime, \ldots ,  \# {\mathcal{I}}_s^\prime\}. \label{eqcirkel} 
\end{equation}
To establish a lower bound on the minimum distance of a code $C$ we
repeat the above process for each $l\in m(C)$. For each such $l$ we 
choose a corresponding $v$, we determine sets ${\mathcal{I}}^{\prime}_i$ as
  above and we calculate the right side of~(\ref{eqcirkel}). The
  smallest value found when $l$ runs through $m(C)$ constitutes a lower bound on the minimum distance.
\end{theorem}
If we compare Theorem~\ref{thestrongone} with Theorem~\ref{the4} we see that
to some extend they have the same flavor. Besides that one deals with
dual codes and the other with primary codes another difference is that
we in Theorem~\ref{thestrongone} has the freedom to choose appropriate sets
${\mathcal{I}}^\prime_0, \ldots ,  {\mathcal{I}}^\prime_s$ whereas the sets
${\mathcal{L}}^\prime(1), \ldots , {\mathcal{L}}^\prime (s+1)$ in Theorem~\ref{the4} are
unique. In~\cite{geilmartin2013further} it was also shown how to lift
Theorem~\ref{thestrongone} to deal with generalised Hamming
weights. Similar remarks as above hold for the two bounds when applied
to such parameters.
\section{A comparison of the new bounds for primary and dual codes}\label{seccomp}
Recall that it was shown in~\cite{agismm} how the Feng-Rao
bound for primary codes and the Feng-Rao bound for dual codes can be
viewed as consequences of each other. This result holds when the bound
is equipped with one of the well-behaving properties WB or OWB. Regarding
the case where WWB is used a possible connection is unknown. In a
similar fashion as the proof in~\cite{agismm} breaks down if one uses
WWB it also breaks down when one tries to prove a correspondence
between Theorem~\ref{the4} and Theorem~\ref{thestrongone}. We
leave it as an open research problem to decide if a general connection exists or
not. \\
In Section~\ref{secfour} we analysed the performance of primary affine
variety codes coming from optimal generalised $C_{ab}$ polynomials. Using
the method from Section~\ref{secsix} one can make a similar analysis
for the corresponding dual codes producing similar code
parameters. As an alternative, below we explain how to derive this result directly from
what we have already shown regarding primary codes from optimal
generalised $C_{ab}$ polynomials.\\
Recall that for optimal generalised
$C_{ab}$ polynomials $\Delta_{\prec_w}(I_q)$ is a box:
$$\Delta_{\prec_w}(I_q)=\{M_1, \ldots ,
M_n\}=\{X^{\alpha_1}Y^{\alpha_2} \mid 0 \leq \alpha_1 < a, 0 \leq
\alpha_2 < q\}.$$
This fact gives us the following crucial implication (as usual we assume $M_1 \prec_w \cdots \prec_w M_n$):
\begin{equation}
M_i=X^{\alpha_1}Y^{\alpha_2} \Rightarrow
M_{n-i+1}=X^{a-1-\alpha_1}Y^{q-1-\alpha_2}, {\mbox{ for }} i=1, \ldots
, n.\label{eqimp}
\end{equation}
Consider codewords
$\vec{c}={\mbox{ev}}\big( \sum_{s=1}^i a_s M_s+I_q\big)$,
$a_s \in {\mathbb{F}}_q$, $a_i \neq 0$, and $\vec{c}^\prime
\in {\mathbb{F}}_q^n$ such that $m(\vec{c}^\prime)=n-i+1$. Let $v$ be
an integer, $0 \leq v < i$. Recall that in Section~\ref{secfour} we
determined ${\mathcal{L}}(u)$, $u=1, \ldots , v+1$. If we use
Theorem~\ref{thestrongone} with $\{l+1, \ldots ,
l+v\}=\{l_1^\prime,\ldots ,
l_s^\prime\}$ (no {\it{a priori}} knowledge) then we can choose 
$${\mathcal{I}}_0^\prime=\{n-l+1 \mid M_l \in {\mathcal{L}}(v+1)\}$$
and for $u=1, \ldots , v$
$${\mathcal{I}}_u^\prime= \{ n-l+1 \mid M_l \in {\mathcal{L}}(u) \}.$$
For $S \subseteq \{1,\ldots , n\}$ define $\bar{S}=\{1, \ldots , n\}
\backslash \{n-s+1 \mid s \in S\}$. Consider 
$$L={\mbox{Span}}_{\mathbb{F}_q}\{{\mbox{ev}}(M_s+I_q) \mid s \in
S\},$$
$$\bar{L}={\mbox{Span}}_{\mathbb{F}_q}\{{\mbox{ev}}(M_s+I_q) \mid s \in
\bar{S} \}.$$ 
As $\#{\mathcal{I}}_0^\prime=\# {\mathcal{L}}(v+1)$ and for $u=1,
\ldots , v$,  $\#{\mathcal{I}}_u^\prime = \#{\mathcal{L}}(u)$ we
conclude that 
Theorem~\ref{thestrongone} produces the same estimate for the minimum
distance of $C^\perp(I,\bar{L})$ as Theorem~\ref{thenewbound} produces for
the minimum distance of 
$C(I,L)$. However, we do not in general have
$C(I,L)=C^\perp(I,\bar{L})$ and therefore the above analysis does not
imply that Theorem~\ref{thenewbound} is a consequence of
Theorem~\ref{thestrongone} even in the case of optimal generalised
$C_{ab}$ polynomials.\\
The above correspondence regarding the minimum distance immediately
carries over to the generalised Hamming weights. In~\cite[Sec.\
4]{geilmartin2013further} we implemented the enhancement of
Theorem~\ref{thestrongone} to generalised Hamming weights for a couple
of concrete dual affine variety codes coming from optimal generalised
$C_{ab}$ polynomials. As a consequence of~(\ref{eqimp}) the estimates found
in~\cite[Sec.\ 4]{geilmartin2013further} for $C^\perp(I,\bar{L})$ also hold for
  $C(I,L)$. This demonstrates the usefulness of the method described in
  Section~\ref{secfourandahalf}.\\

We conclude the section by demonstrating that $d\big( C(I,L) \big) =d
\big( C^\perp (I, \bar{L})\big)$ does not hold for all generalised
$C_{ab}$ polynomials.
\begin{example}\label{exdualprimary}
In this example we consider the generalised $C_{ab}$ polynomial
$F(X,Y)=G(X)-H(Y)\in {\mathbb{F}}_{32}[X,Y]$ where $G(X)=X^{20}+X^{18}+X^{10}+X^9+X^5$ and
$H(Y)=Y^{26}+Y^{22}+Y^{21}+Y^{13}+Y^{11}$.  Observe that both $G$ and
$H$ are 
$({\mathbb{F}}_{32},{\mathbb{F}}_2)$-polynomials and that $G$ satisfies
the condition in Proposition~\ref{proptracelike} ensuring that for each $\eta \in {\mathbb{F}}_2$ there exists exactly $2^{4}=16$
$\gamma \in {\mathbb{F}}_{32}$ such that $G(\gamma)=\eta$. In
particular $F(X,Y)$ has exactly $512$ zeros in ${\mathbb{F}}_{32}$. 
As
$a=\deg G > 16$ $\{ F(X,Y),X^{32}-X,Y^{32}-Y\}$ cannot be a
Gr\"{o}bner basis with respect to $\prec_w$ (it would violate the
footprint bound, Corollary~\ref{thefoot}). Applying Buchberger's
algorithm we find a Gr\"{o}bner basis and from that the corresponding
footprint
\begin{eqnarray}
\Delta_{\prec_w}(I_{32})&=& \{X^{\alpha_1} X^{\alpha_2} \mid  0 \leq
                           \alpha_1 < 12, 0 \leq \alpha_2 < 32\} \nonumber \\
&&\cup {\mbox{ }} \{X^{\alpha_1} X^{\alpha_2} \mid  12 \leq \alpha_1 < 20, 0 \leq
\alpha_2 < 16\} .\nonumber
\end{eqnarray}
Recall the improved construction $\tilde{E}_{imp}(\delta)$ of primary
affine variety codes as introduced in
Definition~\ref{defimpcode}. In a similar way, as Theorem~\ref{thenewbound}
gives rise to the above Feng-Rao style improved primary codes,
Theorem~\ref{thestrongone} gives rise to improved dual codes. These
codes were named $\tilde{C}_{fim}(\delta)$  
in~\cite[Rem.\ 20]{geilmartin2013further}. In a computer experiment we
calculated the parameters of these codes. In Figure~\ref{figfaktisk}
we plot the derived relative parameters. As is seen for some
designed distances $\delta$, $\tilde{E}_{imp}(\delta)$ has the highest
dimension. For other designed distances $\delta$,
$\tilde{C}_{fim}(\delta)$ is of highest dimension.
\begin{figure}
\begin{center}
	\includegraphics[width=0.65\textwidth]{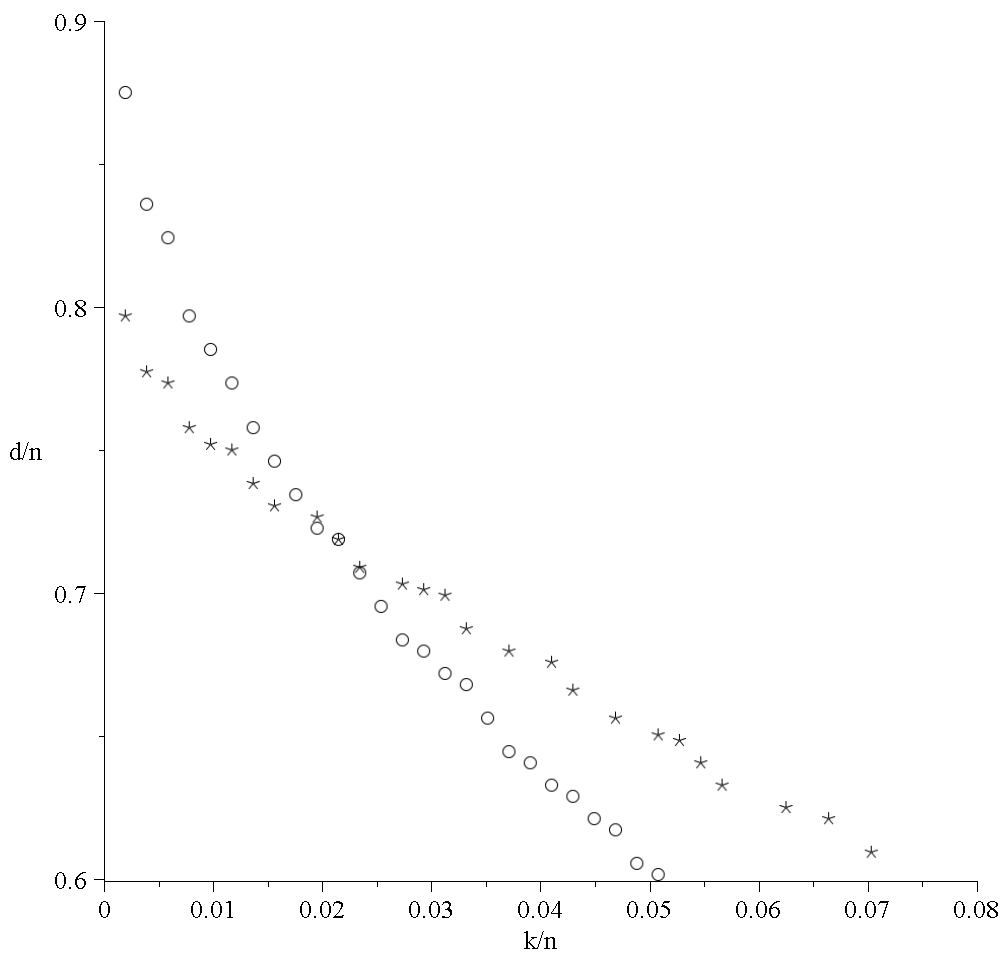}
\end{center}
\caption{Improved codes from Example~\ref{exdualprimary}. A $\circ$ corresponds
  to $\tilde{E}_{imp}(\delta)$, and an $\ast$ corresponds to $\tilde{C}_{fim}(\delta)$.}
\label{figfaktisk}
\end{figure}
\end{example}
\section{Conclusion}\label{seceight}
In this paper we proposed a new bound for the minimum distance and the
generalised Hamming weights of general linear code for which a
generator matrix is known. We demonstrated the usefulness of our bound
in the case of affine variety codes where only the first of the two
order domain conditions is satisfied. For this purpose we introduced
the concept of 
generalised $C_{ab}$ polynomials.
We touched upon the connection to
a bound for dual codes introduced in the recent
paper~\cite{geilmartin2013further}, but leave an investigation of a possible
general relation between the two bounds for further research. It is an
interesting question if there exists examples where our new method
improves on the Feng-Rao bound for one-point algebraic geometric
codes. This would require that we do not choose $v$ as in Remark~\ref{remtksh2} and that we make extensive use of the polynomials $X_i^q-X_i$. Also this question is left for further research. The usual
Feng-Rao bound for primary codes comes with a decoding algorithm that
corrects up to half the estimated minimum
distance~\cite{agismm}. This result holds when the bound is equipped
with the well-behaving property WB. For the case of WWB or OWB no
decoding algorithm is known. Finding a decoding algorithm that
corrects up to half the value guaranteed by Theorem~\ref{thenewbound}
would impose the missing decoding algorithms mentioned above.\\
Part of this research
was done while the second listed author was visiting East China
Normal University. We are grateful to Professor Hao Chen for his hospitality. The authors also gratefully acknowledge the support from
the Danish National Research Foundation and the National Science
Foundation of China (Grant No.\ 11061130539) for the Danish-Chinese
Center for Applications of Algebraic Geometry in Coding Theory and
Cryptography. The authors would like to thank Diego Ruano, Peter Beelen and Ryutaroh Matsumoto for pleasant discussions.

\bibliography{bibfile}
\bibliographystyle{plain}

\end{document}